\documentclass{amsart}
\usepackage{amssymb}
\usepackage{amsmath}
\usepackage{amsfonts}
\usepackage{amsthm}
\usepackage{graphicx}
\usepackage{epsfig}
\usepackage{longtable}

\newtheorem{thm}{Theorem}[section]
\newtheorem{lem}[thm]{Lemma}
\newtheorem{cor}[thm]{Corollary}
\newtheorem{prop}[thm]{Proposition}

\theoremstyle{definition}

\newtheorem{eg}[thm]{Example}

\theoremstyle{remark}
\newtheorem{rmk}[thm]{Remark}

\numberwithin{equation}{section}

\newcommand{\thet}{\vartheta}

\newcommand{\DEF}{{:=}}
\newcommand{\FED}{{=:}}
\newcommand{\weakstarto}{\stackrel{*}\rightarrow}

\newcommand{\ii}{\mathrm{i}}

\newcommand{\Cset}{\mathbb{C}}

\newcommand{\Rset}{\mathbb{R}}

\newcommand{\Sp}{\mathbb{S}}
\newcommand{\Hplus}{\mathbb{H}^{+}}

\newcommand{\PT}[1]{\mathbf{#1}}

\newcommand{\re}{\mathop{\mathrm{Re}}}
\newcommand{\im}{\mathop{\mathrm{Im}}}
\DeclareMathOperator{\dd}{\mathrm{d}}

\DeclareMathOperator{\CAP}{cap}

\DeclareMathOperator{\gammafcn}{\Gamma}

\DeclareMathOperator{\EllipticK}{K}

\DeclareMathOperator{\LegendreP}{P}

\DeclareMathOperator{\POTW}{W}

\DeclareMathOperator{\supp}{supp}

\DeclareMathOperator{\zetafcn}{\zeta}

\DeclareMathOperator{\HyperF}{F}
\newcommand{\Hypergeom}[5]{{\sideset{_#1}{_#2}\HyperF\!\left(\substack{\displaystyle#3\\\displaystyle#4};#5\right)}}

\newcommand{\Pochhsymb}[2]{{\left(#1\right)_{#2}}}

\allowdisplaybreaks[1]

\title[The support of the limit distribution \dots]{The support of the limit distribution of optimal Riesz energy points on sets of revolution in $\mathbb{R}^{3}$} 
\author{ J. S. Brauchart\textasteriskcentered, D. P. Hardin\textdagger, and E. B. Saff\textdaggerdbl } 
\thanks{\noindent \textasteriskcentered  The research of this author was supported, in part,
by the U. S. National Science Foundation under grant DMS-0532154  (D. P. Hardin  and E. B. Saff principal investigators).\\
\textdagger The research of this author was supported, in part,
by the U. S. National Science Foundation under grants DMS-0505756 and DMS-0532154. \\
\textdaggerdbl The research of this author  was supported, in part,
by the U. S. National Science Foundation under grants DMS-0532154 and DMS-0603828. 
}

\begin{document}

\address{J. S. Brauchart, D. P. Hardin and E. B. Saff: 
Center for Constructive Approximation, 
Department of Mathematics, 
Vanderbilt University, 
Nashville, TN 37240, 
USA }
\email{Johann.Brauchart@Vanderbilt.Edu}
\email{Doug.Hardin@Vanderbilt.Edu}
\email{Edward.B.Saff@Vanderbilt.Edu}

\begin{abstract} 
Let $A$ be a compact set in the right-half plane and $\Gamma(A)$ the set in $\mathbb{R}^{3}$ obtained by rotating $A$ about the vertical axis. We investigate the support of the limit distribution of minimal energy point charges on $\Gamma(A)$ that interact according to the Riesz potential $1/r^{s}$, $0<s<1$, where $r$ is the Euclidean distance between points. Potential theory yields that this limit distribution coincides with the equilibrium measure on $\Gamma(A)$ which is supported on the outer boundary of $\Gamma(A)$. We show that there are sets of revolution $\Gamma(A)$ such that the support of the equilibrium measure on $\Gamma(A)$ is {\bf not} the complete outer boundary, in contrast to the Coulomb case $s=1$. However, the support of the limit distribution on the set of revolution $\Gamma(R+A)$ as $R$ goes to infinity, is the full outer boundary for certain sets $A$, in contrast to the logarithmic case ($s=0$).

%
%
%
\end{abstract}


\maketitle


\section{Introduction}

The {\em discrete energy problem} for Riesz kernels $k_{s}(\PT{x})\DEF|\PT{x}|^{-s}$, $s>0$, on compact sets $K$ in $\Rset^3$ is concerned with finding $N$-point systems in $K$ in the most-stable equilibrium; that is, that minimize the $s$-energy
\begin{equation}\label{Esdef}
E_s(X_N):=\sum_{j\neq k} \frac{1}{|\PT{x}_{j}-\PT{x}_{k}|^{s}} = \sum_{k=1}^N\sum_{\begin{subarray}{c} j=1 \\ j\neq k\end{subarray}}^N \frac{1}{|\PT{x}_{j}-\PT{x}_{k}|^{s}}, \qquad s>0,
\end{equation}
among all $N$-point sets $X_N:= \{\PT{x}_1,\dots,\PT{x}_N\} \subset K$, where $|\cdot |$ denotes Euclidean distance. The existence of such configurations follows from both the lower semi-continuity of the Riesz kernel $k_{s}$, $s>0$, and the compactness of $K$. Even in the case that $K$ is the unit sphere in $\Rset^3$, explicit examples of such point sets are known only for a few values of $N$.  For approximate physical models of configurations of minimal energy points for large $N$ on the sphere as well as  toroidal surfaces, see \cite{bowick/etal:2002, bowick/etal:2004}.

The $N$-point system $X_N$  defines a discrete measure $\mu(X_N) \DEF \linebreak[1] (1/N) \sum_{\PT{x}\in X_N} \delta_{\PT{x}}$, by placing the charge $1/N$ at every point $\PT{x}\in X_N$. In this paper we investigate the support of the limit distribution (limit in the weak-star sense as $N\to\infty$) of a sequence of measures $\mu(X_{N}^{*})$, $N\geq2$, induced by minimal energy point configurations $X_{N}^{*}$ on sets of revolution $\Gamma(A)$ in $\Rset^{3}$ obtained by revolving a compact set $A$ in the right-half plane about the vertical axis.

If $0 < s < \dim \Gamma(A)$ (the Hausdorff dimension of $\Gamma(A)$), classical potential theory for the Riesz kernel $k_{s}$ (cf. \cite{landkof:72}) can be used to study this problem. In this case, the limit distribution (as $N\to\infty$) of optimal $N$-point configurations is given by the {\em equilibrium measure} $\mu_{s,\Gamma(A)}$ that uniquely minimizes the continuous energy
\begin{equation*}
\mathcal{I}_s[\mu]\DEF\iint k_{s}(\PT{x}-\PT{y}) \dd\mu(\PT{x}) \dd\mu(\PT{y})
\end{equation*}
over the class $\mathcal{M}(\Gamma(A))$ of (Radon) probability measures $\mu$ supported on $\Gamma(A)$. (For example, when $\Gamma(A)$ is the unit sphere $S^2$ in $\Rset^{3}$ the equilibrium measure is the normalized surface area measure on $S^2$.)

 The   probability measure $\mu_{s,\Gamma(A)}$  is characterized by the following {\em variational principle} \cite[Ch.~II]{landkof:72}: For $\Gamma(A)$ there exists a constant $V_{s}=V_{s}(\Gamma(A))$ such that
\begin{align}
U_{s}^{\mu_{s,\Gamma(A)}} &\geq V_{s} \qquad \text{``approximately everywhere'' on $\Gamma(A)$,} \label{int.equ1} \\
U_{s}^{\mu_{s,\Gamma(A)}} &\leq  V_{s} \qquad \text{ everywhere on the support of $\mu_{s,\Gamma(A)}$.} \label{int.equ2}
\end{align}
Here $U_{s}^{\mu_{s,\Gamma(A)}}$ denotes the {\em equilibrium potential} 
\begin{equation*}
U_{s}^{\mu_{s,\Gamma(A)}}(\PT{x}) \DEF \int k_{s}(\PT{x}-\PT{y}) \, \dd\mu_{s,\Gamma(A)}(\PT{y}), \qquad \PT{x}\in\Rset^{3}.
\end{equation*}
The constant $V_{s}$ is the infimum of the energies of (Radon) probability measures supported on $\Gamma(A)$, that is $V_s=\mathcal{I}_s[\mu_{s,\Gamma(A)}]$. The reciprocal of $V_{s}$ is called the {\em $s$-capacity} of the set $\Gamma(A)$, it is denoted by $\CAP_s\Gamma(A)$.  The term ``approximately everywhere'' means that the property holds everywhere with the possible exception of a set of 
$s$-capacity zero.  It  follows from  \eqref{int.equ1} and \eqref{int.equ2} that $U_{s}^{\mu_{s,\Gamma(A)}}  =  V_{s}$  approximately everywhere on the support of $\mu_{s,\Gamma(A)}$, which provides an integral equation for the equilibrium measure on its support. Knowing this support is therefore an important step in the determination of $\mu_{s,\Gamma(A)}$. 

We remark that Fabrikant et al. \cite{fabrikant/sankar/swamy:1984} provide a method for finding the density $\rho$ of a {\em signed} charge distribution for a prescribed $k_s$-potential distribution   on certain surfaces of revolution in $\Rset^{3}$.  However, their methods  do not apply, for example, to the torus and, more importantly, the distribution they obtain need not be nonnegative. For the analysis of charge distributions in the Coulomb case ($s=1$) on circular or ellipsoidal ``slender toroidal surfaces'', see Cade \cite{cade:1978} and Shail \cite{shail:1979}.

Several important properties of the Riesz equilibrium measure $\mu_{s,K}$ for a compact set $K$ of positive $s$-capacity are summarized in the previously cited book of Landkof. Adopting the same notation, we let $G_{\infty}$ denote the unbounded connected component of the complement of $K$. The boundary $S$ of $G_{\infty}$ is called the {\em outer boundary} of $K$. Furthermore, let $\check{K}$ be ``the set of all points of $K$ each neighborhood of which intersects $K$ in a set of positive $s$-capacity'' (\cite[Ch.~II, no.~13]{landkof:72}). In the case $1\leq s<\dim K$, the First Maximum Principle yields that $\supp\mu_{s,K}\supset\check{S}$. In particular, if $s=1$, then $\supp\mu_{s,K}=\check{S}$. For $s\leq1$, it follows from the superharmonicity of the kernel $k_{s}$   that the equilibrium measure is concentrated on the outer boundary $S$ of $K$. 
In \cite{hardin/saff/stahl:2007} Hardin, Saff, and Stahl proved a stronger result for the logarithmic case (limit as $s\to0^{+}$): For any compact set $A$ in the interior of the right half-plane $\Hplus$, the limit distribution of minimal energy point charges on $\Gamma(A)$ that interact through a logarithmic potential $\log(1/|\PT{x}-\PT{y}|)$ is supported on its ``outer-most'' portion only. The ``outer-most'' part of a torus, for example, is the set of revolution generated by rotating the right semi-circle about the vertical axis. Numerical experiments (cf. \cite{womersley:link} and Section \ref{sec:num.experiments}) suggest that the support of the $s$-equilibrium measure on a torus is, for sufficiently small positive $s$,  likewise a proper subset of the  torus.

In this paper we provide sufficient conditions     under which the support of the equilibrium measure $\mu_{s,\Gamma(A)}$  is a proper subset of the outer boundary of $\Gamma(A)$.    More specifically we show the following.

\begin{itemize}

\item Using rotational symmetry, we demonstrate how to reduce the problem of finding the support of the equilibrium measure $\mu_{s,\Gamma(A)}$ on $\Gamma(A)$ for the (singular) kernel $k_{s}(\PT{x})=1/|\PT{x}|^{s}$   to the  problem of finding the support of the equilibrium measure $\lambda_{s,A}$ on $A$ for a related kernel $\mathcal{K}_{s}$ which is continuous when $0<s<1$ and is singular when $s\ge 1$.
Lemma \ref{lem:K.s.properties} summarizes properties of the kernel $\mathcal{K}_{s}$. 
We further discuss the asymptotics of   optimal $\mathcal{K}_{s}$-energy point configurations on $A$ in both the continuous and singular cases.  \\

\item We show that there are infinite compact sets $A$ for which the support of the equilibrium measure on $\Gamma(A)$ is all of $\Gamma(A)$ for every $0<s<1$. For example, this holds for compact subsets $A$ of a horizontal or a vertical line-segment (see Corollary \ref{cor:horiz.vert}).\\

\item We construct sets of revolution $\Gamma(A)$ such that the support of the equilibrium measure on $\Gamma(A)$ is a proper subset of the outer boundary of $\Gamma(A)$, in contrast to the Coulomb case $s=1$. We demonstrate this for $0<s<1/3$. (This follows from Theorem \ref{3pTHM}.) An example is the outer boundary of the ``washer'' $\Gamma(A)$, where $A$ is the rectangle with lower left corner $1/2-\ii/2$ and upper right corner $1+\ii/2$ (cf. Example \ref{eg:circle}). {\em We conjecture that there exists for every $0<s<1$ a compact set $A$ for which $\supp\mu_{s,\Gamma(A)}$ is a proper subset of the outer boundary of $\Gamma(A)$.} \\

\item We show that for certain sets $A$ the support of the limit distribution on sets of revolution $\Gamma(R+A)$, for the translate $R+A=\{R+z\mid z\in A\}$, tends to the full outer boundary of $\Gamma(R+A)$ as $R\to\infty$. For example, this property holds if the outer boundary of $A$ is a compact subset of a circle with radius $r$ centered at $a>r$ and $0<s<1$ (cf. Lemma \ref{circle}).\\

\item We also show that the support of the equilibrium measure for the logarithmic case ($s=0$) can differ significantly from the case $s>0$. For example, let $A$ be a horizontal line-segment in $\Hplus$. Then we show that $\supp\mu_{s,\Gamma(A)}=\Gamma(A)$ for all $0<s<1$, while it is known that $\supp\mu_{0,\Gamma(A)}$ is the circle generated by the ``right-most'' point of $A$. 
(For further discussion, see end of this section.)
\end{itemize}

{\bf Outline of the paper.} 
In Section \ref{sec:reduction} we reduce the equilibrium problem to a minimal energy problem in the plane with respect to a new kernel $\mathcal{K}_s$ for which we find an explicit expression.

Section \ref{sec:nature} is devoted to the study of $\supp\lambda_{s,A}$ for the kernel $\mathcal{K}_{s}$. A convexity argument (Theorem \ref{thm:convexity}) yields that compact subsets $A$ of horizontal or vertical line-segments are examples with $\supp\lambda_{s,A}=A$ for every $0<s<1$ (Corollary \ref{cor:horiz.vert}). In contrast, we prove the existence of compact sets $A$ for which $\supp\lambda_{s,A}$ is not all of the outer boundary of $A$ by using the variational inequalities for $\mathcal{K}_{s}$. The essential result here is the $3$-point Theorem (Theorem \ref{3pTHM}) which provides a sufficient condition for a point on the outer boundary to not belong to the support of the equilibrium measure corresponding to $\mathcal{K}_s$.

In Section \ref{K.s:R.to.infty} we study the $\mathcal{K}_{s}$-equilibrium measure on sets obtained by translating a given set $A\subset\Hplus$ a distance $R$ units to the right. The asymptotic  expansion of $\mathcal{K}_{s}(R+z,R+w)$, $z,w\in A$, as $R$ becomes large, is given in Lemma \ref{asympt.expansion} and it is sensitive to the order of the limit processes $s\to0^{+}$ and $R\to\infty$. The relation between the energy problem for $\mathcal{K}_{s}$ on $A$ and the energy problem for $\mathcal{K}_{s}$ on the translate $R+A$ is discussed. 

In Section \ref{sec:K.s.infty} we study the kernel that arises as $R\to\infty$, namely
\begin{equation*}
\mathcal{K}_{s}^{(\infty)}(z,w) = - \frac{1}{1-s} \frac{\gammafcn((1+s)/2)}{\sqrt{\pi} \gammafcn(s/2)} \left|z-w\right|^{1-s}, \qquad 0<s<1.
\end{equation*}
We show that any compact subset $A$ of a line-segment $[z^{\prime},z^{\prime\prime}]\subset \Hplus$ has the property that $\supp\lambda_{s,A}^{\infty}=A$ for every $0<s<1$, where $\lambda_{s,A}^\infty$ defines the equilibrium measure for this kernel, and we find an explicit expression for the equilibrium measure $\lambda_{s,A}^{\infty}$ on $A=[z^{\prime},z^{\prime\prime}]$. In case that the outer boundary $S$ of $A$ is a subset of a circle $C$ we get $\supp\lambda_{s,A}^{\infty}=S$ for every $0<s<1$; see Lemma \ref{circle}. In particular, if $S=C$, the equilibrium measure on $A$ for the infinity kernel is simply the normalized arc-length measure on $C$. 

In Section \ref{sec:num.experiments} we discuss the discrete Riesz $s$-energy problem on $\Gamma(A)\subset \Rset^3$ as well as the discrete $\mathcal{K}$-energy problem  on $A\subset\Hplus$ for the kernel $\mathcal{K}=\mathcal{K}_{s}$, $\mathcal{K}=\mathcal{K}_{s}^{(R)}$, and $\mathcal{K}=\mathcal{K}_{s}^{(\infty)}$. We consider the potential theoretical case $0<s<\dim\Gamma(A)$ and the hypersingular case $s\geq\dim\Gamma(A)$. In the hypersingular case the discrete energy problem becomes a weighted energy problem which allows us to use results from \cite{borodachov/hardin/saff:weighted}. We find the limit distribution of minimal $\mathcal{K}$-energy $N$-point systems, consider the separation of such optimal point configurations, and give asymptotics for the discrete minimal energy as $N\to\infty$. Also included are numerical experiments showing minimal energy point configurations on Cassinian ovals, line-segments, and circles.  

An appendix to the paper provides the computations showing convexity of the kernel $\mathcal{K}_{s}$ on the vertical line-segment.

\section{Reduction to the plane, the kernel $\mathcal{K}_{s}$}
\label{sec:reduction}

First we fix some notation. The axis of revolution is identified with the $y$-axis in $\Rset^{3}$. Any vertical cutting plane gives a cross-section of the set of revolution and may serve as a reference plane. Selecting a vertical cutting plane we choose one of the two closed halfplanes and call it $\Hplus$. It may be identified with the complex right half-plane. Then the set of revolution generated by $A\subset\Hplus\DEF\{x+\ii y \mid x\geq0, y\in\Rset \}$ is the set
\begin{equation}
\Gamma(A) \DEF \left\{ \mathbf{R}_{\phi} \PT{x} \mid \PT{x}\in A, 0\leq\phi<2\pi \right\}, \end{equation}
where $\mathbf{R}_{\phi}$ is a rotation by angle $\phi$ about the axis of revolution. The set $\Gamma(A)$ is obtained by revolving $A$ around the vertical axis. Thus, a single point $x+\ii y\in\Hplus$, $x>0$, becomes a horizontal circle with center on the vertical axis.

A Borel measure $\hat{\mu}\in\mathcal{M}(\Rset^{3})$ is {\em rotationally symmetric about the $y$-axis} if 
\begin{equation}
\hat{\mu}(\mathbf{R}_{\phi} B) = \hat{\mu}(B)
\end{equation}
for all Borel sets $B\subset\Rset^{3}$ and for all rotations $\mathbf{R}_{\phi}$ about the $y$-axis. (Here $\mathbf{R}_{\phi} B$ denotes the pointwise rotated set $\{ \mathbf{R}_{\phi} \PT{x} \mid \PT{x}\in B\}$.) 

If $\hat{\mu}\in\mathcal{M}(\Rset^{3})$ is rotationally symmetric about the $y$-axis, then $\hat{\mu}$ can be written as a product of two measures, the normalized Lebesgue measure on the half-open interval $[0,2\pi)$ and a measure $\mu$ on $\Hplus$, that is
\begin{equation}
\dd\hat{\mu} = \frac{\dd \phi}{2\pi} \, \dd\mu, \qquad \mu = \hat{\mu} \circ \Gamma \in \mathcal{M}(\Hplus).
\end{equation}
Then the energy of the (compactly supported) measure $\hat{\mu}$ can be expressed as
\begin{equation}
\begin{split} \label{connecting.energies}
\mathcal{I}_{s}[\hat{\mu}] 
&= \iint_{\Rset^{3}\times\Rset^{3}} k_{s}(\PT{x}-\PT{y}) \, \dd\hat{\mu}(\PT{x}) \dd\hat{\mu}(\PT{y}) \\
&= \iint_{\Hplus\times\Hplus} \mathcal{K}_{s}(z,w) \, \dd\mu(z) \dd\mu(w) \FED \mathcal{J}_{s}[\mu],
\end{split}
\end{equation}
where the kernel $\mathcal{K}_{s}(z,w)$ is given by the integral 
\begin{equation}
\label{eq:kernel.K}
\mathcal{K}_{s}(z,w) \DEF \frac{1}{2\pi} \int_{0}^{2\pi} \frac{1}{\left|\PT{R}_{\phi} z - w\right|^{s}} \, \dd\phi.
\end{equation}

\subsection{The energy problem for $\mathcal{K}_{s}$} 

Let $0<s<1$. Let $A\subset\Hplus$ be a compact set such that $\CAP_{s}\Gamma(A)>0$. Then the uniqueness of the equilibrium measure $\mu_{s,\Gamma(A)}$ on $\Gamma(A)$ and the symmetry of the revolved set $\Gamma(A)$ imply that $\mu_{s,\Gamma(A)}$ is rotationally symmetric about the $y$-axis and so $\dd\mu_{s,\Gamma(A)}=[\dd\phi/(2\pi)]\dd\lambda_{s,A}$, where $\lambda_{s,A}=\mu_{s,\Gamma(A)}\circ \Gamma \in \mathcal{M}(\Hplus)$. Furthermore, if $\nu\in\mathcal{M}(\Hplus)$, then $\dd\hat{\nu} \DEF [\dd\phi/(2\pi)]\dd\nu$ is rotationally symmetric about the $y$-axis and so we have
\begin{equation*}
\mathcal{J}_{s}[\lambda_{s,A}] \geq \inf_{\nu\in\mathcal{M}(\Hplus)} \mathcal{J}_{s}[\nu] = \inf_{\nu\in\mathcal{M}(\Hplus)} \mathcal{I}_{s}[\hat{\nu}] \geq \mathcal{I}_{s}[\mu_{s,A}] = \mathcal{J}_{s}[\lambda_{s,A}].
\end{equation*}
In the case $0<s<1$ the equilibrium measure on $\Gamma(A)$ is concentrated on the outer boundary of $\Gamma(A)$ (cf. \cite[Ch.~II, no.~13]{landkof:72}).
\begin{prop}\label{prop.K.s}
Let $0<s<\dim\Gamma(A)$. Let $A\subset\Hplus$ be a compact set with $\CAP_{s} \Gamma(A) > 0$. Then $\lambda_{s,A}=\mu_{s,\Gamma(A)}\circ \Gamma$ uniquely minimizes $\mathcal{J}_{s}[\nu]$ over all measures $\nu\in\mathcal{M}(A)$. Thus, $\lambda_{s,A}$ is the {\em equilibrium measure} on $A$ for the kernel $\mathcal{K}_{s}$. It is supported on the outer boundary of $A$. 
\end{prop}

The $\mathcal{K}_{s}$-energy of a measure was defined in \eqref{connecting.energies}. The energy $V_{\mathcal{K}_{s}}$ of $A$ is given by
\begin{equation}
V_{\mathcal{K}_{s}}(A) \DEF \inf\left\{ \mathcal{J}_{s}[\nu] \mid \nu\in\mathcal{M}(A) \right\}.
\end{equation}
The following relations hold:
\begin{equation}
V_{\mathcal{K}_{s}}(A) = \mathcal{J}_{s}[\lambda_{s,A}] = \mathcal{I}_{s}[\mu_{s,\Gamma(A)}] = V_{s}(\Gamma(A)).
\end{equation}
For $\nu\in\mathcal{M}(A)$, we define the $\mathcal{K}_{s}$-potential $W_{s}^{\nu}$ by
\begin{equation}
W_{s}^{\nu}(z) \DEF \int_{A} \mathcal{K}_{s}(z,w) \dd\nu(w), \qquad z\in\Hplus.
\end{equation}
Let $\hat{\nu}\in\mathcal{M}(\Gamma(A))$ be rotationally symmetric with $\dd\hat{\nu}=[\dd\phi/(2\pi)]\dd\nu$, where $\nu=\hat{\nu}\circ\Gamma\in\mathcal{M}(A)$. Then the potential $U_{s}^{\hat{\nu}}$ is constant on circles $\Gamma(\{z\})$, $z\in\Hplus$. Abusing notation, there holds the following connecting formula
\begin{equation}
\begin{split}
U_{s}^{\hat{\nu}}(z) 
&= \int_{\Gamma(A)} k_{s}(z-\PT{y}) \dd\hat{\nu}(\PT{y}) = \frac{1}{2\pi} \int_{A} \int_{0}^{2\pi} k_{s}(z-\mathbf{R}_{\phi} w) \dd\phi \dd\nu(w) \\
&= \int_{A} \mathcal{K}_{s}(z,w) \dd\nu(w) = W_{s}^{\nu}(z), \qquad z\in\Hplus.
\end{split}
\end{equation}
From the properties \eqref{int.equ1} and \eqref{int.equ2} of the equilibrium potential $U_{s}^{\mu_{s,\Gamma(A)}}$ we infer the {\em variational inequalities for $\mathcal{K}_{s}$} for compact sets $A$ in the interior of $\Hplus$: 
\begin{align}
W_{s}^{\lambda_{s,A}} &\geq V_{\mathcal{K}_{s}}(A) \qquad \text{everywhere on $A$,} \label{variat.inequ.a} \\
W_{s}^{\lambda_{s,A}} &\leq V_{\mathcal{K}_{s}}(A) \qquad \text{on $\supp \lambda_{s,A}$.} \label{variat.inequ.b}
\end{align}
In this case we do no longer need an ``approximately everywhere'' exceptional set, since each point of $A$ generates a circle in $\Rset^{3}$ with positive capacity.

\subsection{Properties of the kernel $\mathcal{K}_{s}$}

Let $z\DEF x+\ii y$, $w\DEF u+\ii v$, where $x,y,u,v\in\Rset$. Let $w_{*} \DEF -\overline{w} = -u + \ii v$ denote the reflection of $w$ in the imaginary axis. 
\begin{lem} \label{lem:K.s.properties}
Let $s>0$. The kernel $\mathcal{K}_{s}:\Hplus\times\Hplus\to\Rset$ in \eqref{eq:kernel.K} has the following properties:
\begin{enumerate}
\item $\mathcal{K}_{s}(z,w)$ is  well defined for $z\neq w$ for all $s>0$.
\item $\mathcal{K}_{s}$ is symmetric: $\mathcal{K}_{s}(z,w) = \mathcal{K}_{s}(w,z)$.
\item $\mathcal{K}_{s}$ is homogeneous: $\mathcal{K}_{s}(r z,r w) = r^{-s}\mathcal{K}_{s}(z,w)$ for all $r>0$.
\item $\mathcal{K}_{s}$ is continuous at all points $(z,w)\in\Hplus\times\Hplus$ with $z\neq w$. If $0<s<1$, then $\mathcal{K}_{s}$ is continuous at $(w,w)$ with $\re[w]>0$. $\mathcal{K}_{s}(z,w)$ is singular at $z=w$ for $s\geq1$. 
\item If $w$ is on the imaginary axis and $s>0$, then $\mathcal{K}_{s}(z,w) = |z-w|^{-s}$, $z\neq w$. If $\re[w]>0$, then, for $s>1$, the following limit holds:
\begin{equation} \label{K.s.asympt}
\left|z-w\right|^{s-1}\mathcal{K}_{s}(z,w) \to \frac{\gammafcn((s-1)/2)}{\sqrt{\pi}\gammafcn(s/2)} \frac{1}{\left|w-w_*\right|}, \qquad \text{as $z\to w$.}
\end{equation}
\item $\mathcal{K}_{s}(u+\ii t,u+\ii v)$ decreases along vertical lines as $|t-v|$ grows and $\mathcal{K}_{s}(u+\ii y,u+t+\ii y)$ decreases along horizontal lines as $t>0$ grows\footnote{This follows from differentiating the integral \eqref{eq:kernel.sellip} with respect to $t$.}.
\item \label{Prop.5} Let $0<s<1$. For fixed $w$ with $\re[w]>0$, the function $\mathcal{K}_{s}(z,w)$ has exactly one global maximum at $z=w$ in $\Hplus$. At $(w,w)$ or $(w_{*},w)$, the kernel $\mathcal{K}_{s}$ takes the value
\begin{equation}
\mathcal{K}_{s}(w,w) = \mathcal{K}_{s}(w_{*},w) = \mathcal{I}_{s}\left(\Sp^{1};\frac{\dd\phi}{2\pi}\right) \left|\re w\right|^{-s}, \label{K.s:w.w}
\end{equation}
where
\begin{equation}
\mathcal{I}_{s}\left(\Sp^{1};\frac{\dd\phi}{2\pi}\right) = 2^{-s} \frac{\gammafcn((1-s)/2)}{\sqrt{\pi} \gammafcn(1-s/2)} = \frac{\gammafcn(1-s)}{\left[\gammafcn(1-s/2)\right]^{2}}. \label{I.s}
\end{equation}
\item The kernel $\mathcal{K}_{s}$ has the following representations in terms of hypergeometric functions \cite{andrews/askey/roy:1999} or in terms of a Legendre function \cite{abramowitz/stegun:70} 
\begin{align}
\mathcal{K}_{s}(z,w) 
&= \left|z-w_{*}\right|^{-s} \Hypergeom{2}{1}{s/2,1/2}{1}{1-\frac{\left|z-w\right|^{2}}{\left|z-w_{*}\right|^{2}}} \label{K.s.1st} \\
&= \left\{ \frac{2}{\left|z-w_{*}\right|+\left|z-w\right|} \right\}^{s} \Hypergeom{2}{1}{s/2,s/2}{1}{\left\{\frac{\left|z-w_{*}\right|-\left|z-w\right|}{\left|z-w_{*}\right|+\left|z-w\right|}\right\}^{2}} \label{K.s.2nd}\\
&= \left|z-w_{*}\right|^{-s/2} \left|z-w\right|^{-s/2} \LegendreP_{s/2-1}^{0}\left( \frac{1}{2} \, \frac{\left|z-w\right|}{\left|z-w_{*}\right|} + \frac{1}{2} \,  \frac{\left|z-w_{*}\right|}{\left|z-w\right|} \right). \label{LegendrePrepr}
\end{align}
(Observe that the Legendre function is evaluated at values $>1$ if $\re z>0$ or $\re w>0$.) For $s>1$ one can factor out the singularity at $z=w$,
\begin{equation}
\mathcal{K}_{s}(z,w) = \frac{\left|z-w\right|^{1-s}}{\left|z-w_*\right|} \Hypergeom{2}{1}{1-s/2,1/2}{1}{1-\frac{\left|z-w\right|^{2}}{\left|z-w_{*}\right|^{2}}}, \qquad z\neq w. \label{K.s.3rd}
\end{equation}
\item \label{Prop.7} As $s\to0^{+}$ we recover the logarithmic kernel $\mathcal{K}_{0}$ studied in \cite{hardin/saff/stahl:2007}:
\begin{equation} \label{eq:lim.s.to.0}
\lim_{s\to0^{+}} \frac{\mathcal{K}_{s}(z,w)-1}{s} = \log \frac{2}{\left|z-w\right|+\left|z-w_{*}\right|}.
\end{equation}
\item As $s\to1^{-}$, $\mathcal{K}_{s}(z,w)\to\mathcal{K}_{1}(z,w)$, where
\begin{equation}
\mathcal{K}_{1}(z,w) \DEF \frac{2}{\pi} \, \frac{2}{\left|z-w_{*}\right|+\left|z-w\right|} \EllipticK\left(\left\{\frac{\left|z-w_{*}\right|-\left|z-w\right|}{\left|z-w_{*}\right|+\left|z-w\right|}\right\}^{2}\right), \label{EllipticKrepr}
\end{equation}
and $\EllipticK$ denotes the {\em complete Elliptic integral of the first kind} \cite{abramowitz/stegun:70}. 
\end{enumerate}
\end{lem}

\begin{rmk}
For the special case of the sphere the formula \eqref{K.s.2nd} reduces to the formula (4.14) in Dragnev and Saff \cite{dragnev/saff:2007}.
\end{rmk}

The level sets of $\mathcal{K}_{s}(\cdot,w)$, $w\in\Hplus$ fixed, look like Cassinian ovals, cf. Figure \ref{level.sets.s01div02}. The asymptotical behavior of $\mathcal{K}_{s}(R+z,R+w)$ as $R\to\infty$ is given in Lemma \ref{asympt.expansion}. 
\begin{figure}[h]
\begin{center}
\caption{\label{level.sets.s01div02}Level sets for $\mathcal{K}_{s}(z,1)$, $s=1/2$.}
\end{center}
\end{figure}

\begin{proof}[Proof of Lemma \ref{lem:K.s.properties}]
Let $z,w\in\Hplus$ with $z=x+\ii y$ and $w=u+\ii v$. The relation
\begin{equation*}
\left|\PT{R}_{\phi} z - w\right|^{2} = x^{2} + u^{2} - 2xu\cos\phi + \left(y - v\right)^{2}
\end{equation*}
gives $1/|\PT{R}_{\phi} z - w|^{s} = ( E - F \cos\phi )^{-s/2}$ for the integrand in \eqref{eq:kernel.K}, where we define 
\begin{equation}
E \DEF x^{2} + u^{2} + \left(y - v\right)^{2}, \qquad F \DEF 2xu. \label{def.E.F}
\end{equation}
By \eqref{def.E.F} the kernel $\mathcal{K}_{s}(z,w)$ is symmetric in $z,w$. The substitution $\phi=\psi+\pi$ yields
\begin{equation}
\mathcal{K}_{s}(z,w) = \frac{1}{2\pi} \int_{-\pi}^{\pi} \left( E + F \cos\psi \right)^{-s/2} \dd\psi = \frac{1}{\pi} \int_{0}^{\pi} \left( E + F \cos\phi \right)^{-s/2} \dd\phi. \label{eq:kernel.pre.ellip}
\end{equation}
Applying the half angle formula and substituting $\psi=\phi/2$ we obtain
\begin{equation}
\mathcal{K}_{s}(z,w) 
= \left(E+F\right)^{-s/2} \frac{2}{\pi} \int_{0}^{\pi/2} \left( 1 - \frac{2F}{E + F} \sin^{2}\psi\right)^{-s/2} \dd\psi. \label{eq:kernel.sellip}
\end{equation}
The integral in \eqref{eq:kernel.sellip} resembles that of a complete elliptic integral. Indeed, for $s=1$ this integral is the {\em complete Elliptic integral of the first kind} $\EllipticK(k^2)$ with elliptic modulus $k^2=2F/(E+F)$. (See for example \cite[17.2.19,17.3.1]{abramowitz/stegun:70}.) The transformation $\phi=\psi-\pi$ in \eqref{eq:kernel.pre.ellip} gives
\begin{equation}
\mathcal{K}_{s}(z,w) = E^{-s/2} \frac{1}{\pi} \int_{0}^{\pi} \left( 1 - \frac{F}{E} \cos\psi \right)^{-s/2} \dd\psi. \label{gen.Epstein}
\end{equation}
The integral in \eqref{gen.Epstein} is a generalization of Epstein and Hubbells elliptic integral. We refer to \cite{Srivastava/Siddiqi:1995} for a discussion of these elliptic-type integrals.

A change of variables $t=\sin^{2}\psi$ in \eqref{eq:kernel.sellip} yields
\begin{equation} \label{K.s:int.repr}
\mathcal{K}_{s}(z,w) 
= \left( E + F \right)^{-s/2} \frac{1}{\pi} \int_{0}^{1} t^{1/2-1} \left(1-t\right)^{1/2-1} \left( 1 - \frac{2F}{E+F} t \right)^{-s/2} \dd t. 
\end{equation}
Recall, that the {\em Gauss hypergeometric series} \cite[15.1.1]{abramowitz/stegun:70}
\begin{equation*}
\sum_{k=0}^{\infty}\frac{\Pochhsymb{a}{k}\Pochhsymb{b}{k}}{\Pochhsymb{c}{k}}\,\frac{z^{k}}{k!}=\frac{\gammafcn(c)}{\gammafcn(a)\gammafcn(b)}\sum_{k=0}^{\infty}\frac{\gammafcn(a+k)\gammafcn(b+k)}{\gammafcn(c+k)}\,\frac{z^{k}}{k!}
\end{equation*}
represents the {\em Gauss hypergeometric function} $\Hypergeom{2}{1}{a,b}{c}{z}$ for all complex $z$ within the circle of convergence, the unit circle $|z|=1$. The analytic continuation in the $z$-plane cut along the segment $[1,\infty]$, \cite[15.3.1]{abramowitz/stegun:70},
\begin{equation*} 
\frac{\gammafcn(c)}{\gammafcn(b)\gammafcn(c-b)}\int_{0}^{1} t^{b-1}\left(1-t\right)^{c-b-1}\left(1-z\,t\right)^{-a} \dd{t},\quad \re c>\re b>0,
\end{equation*}
can be used to derive a {\em hypergeometric function representation} of the kernel $\mathcal{K}_{s}(z,w)$,
\begin{equation}
\mathcal{K}_{s}(z,w) = \left( E + F \right)^{-s/2} \Hypergeom{2}{1}{s/2,1/2}{1}{\frac{2F}{E+F}}.\label{eq:K.formB}
\end{equation}
Let $w_{*} \DEF -\overline{w} = -u + \ii v$ denote the reflection of $w$ in the imaginary axis. Then
\begin{equation}
\begin{split}
E-F &= \left(x-u\right)^{2} + \left(y - v\right)^{2} = \left|z-w\right|^{2}, \\
E+F &= \left(x+u\right)^{2} + \left(y - v\right)^{2} = \left|z-w_{*}\right|^{2}, \label{eq:E.pm.F}
\end{split}
\end{equation}
and we get the relations
\begin{align*}
0 \leq \frac{2F}{E+F} &= \frac{\left|z-w_{*}\right|^{2}-\left|z-w\right|^{2}}{\left|z-w_{*}\right|^{2}} = \frac{4xu}{\left(x+u\right)^{2}+\left(y-v\right)^{2}} \leq 1, \qquad z,w\in\Hplus.
\end{align*}
Substitution of \eqref{eq:E.pm.F} into \eqref{eq:K.formB} yields \eqref{K.s.1st}.

The hypergeometric function in \eqref{K.s.1st} is of the form $\Hypergeom{2}{1}{a,b}{2b}{\zeta}$. The quadratic transformation \cite[15.3.17]{abramowitz/stegun:70} yields a more symmetrical representation \eqref{K.s.2nd}.

In the argument of the hypergeometric function in \eqref{K.s.2nd} appears the expression
\begin{equation*}
\xi \DEF \frac{\left|z-w_{*}\right|-\left|z-w\right|}{\left|z-w_{*}\right|+\left|z-w\right|} = \frac{\left|z-w_{*}\right|^{2}-\left|z-w\right|^{2}}{\left(\left|z-w_{*}\right|+\left|z-w\right|\right)^{2}} = \frac{4 \re z \, \re w}{\left(\left|z-w_{*}\right|+\left|z-w\right|\right)^{2}}.
\end{equation*}
It satisfies $\xi^{2} \leq 1$ and equality holds for $z=w$ or $z=w_{*}$ only. Therefore we may use the series expansion of the hypergeometric function to get
\begin{equation}
\mathcal{K}_{s}(z,w) = \left( \frac{2}{\left|z-w_{*}\right|+\left|z-w\right|} \right)^{s} \sum_{\ell=0}^{\infty} \left( \frac{\Pochhsymb{s/2}{\ell}}{\ell!} \, \xi^{\ell}\right)^{2} , \qquad z\neq w,w_{*}. \label{eq:K.formD}
\end{equation}
If $0<s<1$, this series converges even for $z=w,w_{*}$. For $z=w,w_{*}$ the argument of the hypergeometric function in \eqref{K.s.1st}, \eqref{K.s.2nd} is $1$. From \cite[15.1.20]{abramowitz/stegun:70}
\begin{equation*}
\mathcal{K}_{s}(w,w) = \mathcal{K}_{s}(w_{*},w) = \frac{\gammafcn(1-s)}{\left[\gammafcn(1-s/2)\right]^{2}} \left|\re w\right|^{-s} = 2^{-s} \, \frac{\gammafcn((1-s)/2)}{\sqrt{\pi} \, \gammafcn(1-s/2)} \left|\re w\right|^{-s}.
\end{equation*}
(The first two relations follow from \eqref{K.s.2nd}, the last one from \eqref{K.s.1st}.) 
Note, the leading coefficient at the right-most is the energy $\mathcal{I}_{s}(\Sp^{1};\dd\phi/(2\pi))$, where $\Sp^{1}$ is the unit circle and $\dd\phi/(2\pi)$ the uniform measure on $\Sp^{1}$. This shows \eqref{K.s:w.w} and \eqref{I.s}.

Those hypergeometric functions that allow a quadratic transformation are connected with Legendre functions. From \eqref{K.s.1st} and relation \cite[15.4.7]{abramowitz/stegun:70} 
we get \eqref{LegendrePrepr}.

From \eqref{K.s.1st} and relation \cite[15.3.3]{abramowitz/stegun:70} we get \eqref{K.s.3rd}. From \eqref{K.s.3rd} follows \eqref{K.s.asympt}. If $\re[w]=0$, then $w=w_*$. Hence, by \eqref{K.s.1st}, $\mathcal{K}_s(z,w)=|z-w|^{-s}$, $z\neq w$, for $s>0$.

The {\em complete Elliptic integral of the first kind} $\EllipticK(k^2)$ \cite[17.3.1]{abramowitz/stegun:70} can be represented through a hypergeometric function \cite[17.3.9]{abramowitz/stegun:70},
\begin{equation*}
\EllipticK(k^{2}) = \int_{0}^{\pi/2} \frac{\dd\thet}{\sqrt{1-k^2 \left(\sin\thet\right)^{2}}} = \frac{\pi}{2} \Hypergeom{2}{1}{1/2,1/2}{1}{k^{2}}.
\end{equation*}
Thus \eqref{EllipticKrepr} follows from  \eqref{K.s.2nd}.

As $s\to0^{+}$, the hypergeometric series in \eqref{eq:K.formD} reduces to $1$. Thus it makes sense to consider the quotient $(\mathcal{K}_{s}(z,w)-1)/s$. Fix $z,w\in\Hplus$ in \eqref{eq:K.formD}. Let $z\neq w$. Then 
\begin{equation*}
\begin{split} 
\frac{\dd}{\dd s} \mathcal{K}_{s}(z,w) 
&= \mathcal{K}_{s}(z,w) \log \frac{2}{\left|z-w\right|+\left|z-w_{*}\right|} \\
&\phantom{=\pm}+ \left( \frac{2}{\left|z-w_{*}\right|+\left|z-w\right|} \right)^{s} \frac{\dd}{\dd s} \left[ \sum_{\ell=0}^{\infty} \left( \frac{\Pochhsymb{s/2}{\ell}}{\ell!} \xi^{\ell}\right)^{2} \right].
\end{split}
\end{equation*}
We are only interested in $\dd \mathcal{K}_{s}(z,w)/\dd s$ at $s=0^{+}$. $\mathcal{K}_{s}(z,w)$ becomes one at $s=0$. The derivative in the right-most term above exists and vanishes. This follows from
\begin{equation*}
\frac{1}{s} \left[ \sum_{\ell=0}^{\infty} \left( \frac{\Pochhsymb{s/2}{\ell}}{\ell!} \, \xi^{\ell}\right)^{2} - 1 \right] 
= \frac{1}{s} \left( \frac{s}{2} \xi \right)^{2} \sum_{\ell=0}^{\infty} \left( \frac{\Pochhsymb{1+s/2}{\ell}}{(\ell+1)!} \xi^{\ell} \right)^{2}
\end{equation*}
and the limit process $s\to0^{+}$. By the ratio test, the infinite series on the right-hand side above is absolutely convergent for $|\xi|<1$ (that is $z\neq w$) and $0<s<1$. In the case $z=w$ one uses \eqref{K.s:w.w} instead of \eqref{eq:K.formD}.
\end{proof}

\section{The support of the equilibrium measure for the kernel $\mathcal{K}_{s}$}
\label{sec:nature}

By Proposition \ref{prop.K.s}, the equilibrium measure $\lambda_{s,A}$ on $A$ for $\mathcal{K}_{s}$ is supported on the outer boundary $S$ of $A$. A convexity argument yields sufficient conditions for $\supp\lambda_{s,A}=S$. Recall that a function $f:[a,b]\to\Rset$ is {\em strictly convex on $[a,b]$} if $f( \tau x + ( 1 - \tau ) y) < \tau f(x) + ( 1 - \tau ) f(y)$ for all $a\leq x < y \leq b$ and $0 < \tau < 1$. 
\begin{thm}\label{thm:convexity}
Let $0<s<1$ and $A$ be a compact set in the interior of $\Hplus$. 
\begin{enumerate}
\item[(i)] If $\gamma:[a,b]\to\Hplus$, $a<b$, is a simple continuous {\bf non-closed} curve covering the outer boundary $S$ of $A$, that is $S \subset \gamma^{*} \DEF \left\{ \gamma(t) \mid a\leq t \leq b \right\}$, and $\mathcal{K}_{s}(\gamma(\cdot),\gamma(t))$ is a strictly convex function on the intervals $[a,t]$ and $[t,b]$ for each fixed $t\in[a,b]$, then there is some closed interval $I\subset[a,b]$ such that $\supp\lambda_{s,A}=\gamma(I)\cap S$.
\item[(ii)] If $\gamma:[0,b]\to\Hplus$ is a simple continuous {\bf closed} curve, that is $\gamma(0)=\gamma(b)$, with $S\subset\gamma^{*}$ and extended periodically by $\gamma(t)=\gamma(t+b)$, and $\mathcal{K}_{s}(\gamma(\cdot),\gamma(t))$ is a strictly convex function on the interval $[t,t+b]$ for each fixed $t\in[0,b]$, then $\supp\lambda_{s,A}=S$.
\end{enumerate}
\end{thm}
\begin{rmk}
Note, that $S$ is only required to be a compact subset of $\gamma^{*}$. For example, $S$ may be a Cantor subset of $\gamma^{*}$. 
\end{rmk}

\begin{proof}[Proof of Theorem \ref{thm:convexity}]
Set $\lambda=\lambda_{s,A}$ and $\POTW^{\lambda}=\POTW_{\mathcal{K}_{s}}^{\lambda}$. We have $\supp \lambda\subset S\subset\gamma^{*}$. Suppose G is a component of the complement of $\supp \lambda$ in $\gamma^{*}$. Now observe, that by our assumptions, $G$ always corresponds to a subinterval $I$ of one of the sets $[a,t]$, $[t,b]$ or $[t,t+b]$ for $\gamma(t)\in\supp\lambda$. Two cases are possible: (i) Both boundary points of $G$ are in $\supp\lambda$. Then the equilibrium potential $\POTW^{\lambda}$ assumes the value $\mathcal{J}_{\mathcal{K}_{s}}[\lambda]$ on the boundary of $G$ and, due to strict convexity of $\POTW^{\lambda}\circ\gamma$ on $I$, is strictly less than this value in the open set $G$. Since $\POTW^{\lambda}\geq\mathcal{J}_{\mathcal{K}_{s}}[\lambda]$ on $A\supset S$, no point of $G$ is in $A$. (ii) At least one boundary point of $G$ is not in $\supp\lambda$. This can only happen when $\gamma$ is a {\bf non-closed} curve. Without further assumptions the convexity property alone is insufficient to show $G\cap A=\emptyset$. From (i) follows the existence of some closed interval $I\subset[a,b]$ such that $\supp\lambda=\gamma(I)\cap S$. If $\gamma$ is a closed curve, then $I=[0,b]$.
\end{proof}

\begin{rmk}
In the proof of Theorem \ref{thm:convexity} we use three main properties: (i) The kernel is continuous, (ii) $\supp \lambda_{s,A}\subset S$, and (iii) the equilibrium potential satisfies a variational principle. These properties also hold for $\mathcal{K}_{s}^{(\infty)}$ introduced in Section \ref{sec:K.s.infty}. Therefore, Theorem \ref{thm:convexity} can be applied in case of $\mathcal{K}_{s}^{(\infty)}$.
\end{rmk}

Using Theorem \ref{thm:convexity}(i) we next show that any compact subset $A$ of a horizontal or vertical line-segment satisfies $\supp\lambda_{s,A}=A$ for every $0<s<1$. We contrast this with the logarithmic case, where it is still true that $\supp\lambda_{0,A}=A$ in case of a vertical line-segment \cite[Cor.~1]{hardin/saff/stahl:2007}. However, in case of a horizontal line-segment one has that $\lambda_{0,A}$ is a unit point charge at the right-most point of $A$ \cite[Thm.~1]{hardin/saff/stahl:2007}.
\begin{cor} \label{cor:horiz.vert}
Suppose $A$ is a compact subset of either (a) the horizontal line-segment $[a+\ii c,b+\ii c]$, $0<a<b$, or (b) the vertical line-segment $[R+\ii c,R+\ii d]$, $R>0$, $c<d$. Then $\supp\lambda_{s,A}=A$ for every $0<s<1$.
\end{cor}
\begin{proof}
For (a) consider the parametrization $\gamma(x)=x+\ii c$, $a\leq x\leq b$. From \eqref{K.s.2nd},
\begin{align*}
\mathcal{K}_{s}(\gamma(x),\gamma(u)) &= x^{-s} \Hypergeom{2}{1}{s/2,s/2}{1}{\frac{u^{2}}{x^{2}}} = \sum_{n=0}^{\infty} \frac{\Pochhsymb{s/2}{n}\Pochhsymb{s/2}{n}}{\Pochhsymb{1}{n} n!} u^{2n} x^{-s-2n} , \quad x > u, \\
\mathcal{K}_{s}(\gamma(x),\gamma(u)) &= u^{-s} \Hypergeom{2}{1}{s/2,s/2}{1}{\frac{x^{2}}{u^{2}}} = \sum_{n=0}^{\infty} \frac{\Pochhsymb{s/2}{n}\Pochhsymb{s/2}{n}}{\Pochhsymb{1}{n} n!} x^{2n} u^{-s-2n}, \quad x < u.
\end{align*}
From $(x^{-s-2n})^{\prime\prime}>0$, $n\geq0$, and $(x^{2n})^{\prime\prime}>0$, $n\geq1$, we get $[\mathcal{K}_{s}(\gamma(x),\gamma(u))]^{\prime\prime} > 0$ for $x\neq u$ and for every $0<s<1$. Termwise differentiation is justified by uniform convergence for $|x-u|\geq\delta$. By Theorem \ref{thm:convexity}, $\supp\lambda_{s,A}=\gamma(I)\cap A$ for some $I=[a^{\prime},b^{\prime}]\subset[a,b]$. From the series representations above we observe that the kernel $\mathcal{K}_{s}(\gamma(x),\gamma(u))$ is a strictly increasing function in $x$ for $x<u$ and it is a strictly decreasing function in $x$ for $x>u$. Hence, $\POTW_{\mathcal{K}_{s}}^{\lambda_{s,A}}\circ\gamma<\mathcal{I}_{s}[\lambda_{s,A}]$ on $[a,b]\setminus I$. By variational inequality \eqref{variat.inequ.a}, $I=[a,b]$.

For (b) consider the parametrization $\gamma(y)=R+\ii y$, $c\leq y\leq d$. From \eqref{K.s.1st}, 
\begin{equation}
\mathcal{K}_{s}(\gamma(y),\gamma(v)) = \left[ 4 R^{2} + \left( y - v \right)^{2} \right]^{-s/2} \Hypergeom{2}{1}{s/2,1/2}{1}{\frac{4 R^{2}}{4 R^{2} + \left( y - v \right)^{2}}}. \label{K.s.vertical-line}
\end{equation}
A direct calculation (assisted by Mathematica, see Appendix \ref{appdx} for more details) shows that $\dd^{2} [\mathcal{K}_{s}(\gamma(y),\gamma(v))]/\dd y^{2} > 0$ for every $0<s<1$. By Theorem \ref{thm:convexity}, $\supp\lambda_{s,A}=\gamma(I)\cap A$ for some $I=[c^{\prime},d^{\prime}]\subset[c,d]$. From the representation above we observe that the kernel $\mathcal{K}_{s}(\gamma(y),\gamma(v))$ is a strictly decreasing function in $y$ for growing $|y-v|$. Proceeding as in part (a) we get $I=[c,d]$.
\end{proof}

In contrast to the horizontal or the vertical line-segment we will show that there are compact sets $A$ in the interior of $\Hplus$ for which, in fact, the support of the equilibrium measure on $A$ for $\mathcal{K}_{s}$ is a proper subset of the outer boundary of $A$. 
\begin{eg} \label{eg:circle} Let $A$ be the rectangle with lower left corner $1/2-\ii/2$ and upper right corner $1+\ii/2$. Using Theorem \ref{3pTHM}(c) below with $x=1/2$ and $z^{\prime}=1+\ii/2$, it follows that $1/2\notin\supp\lambda_{s,A}$ for $0<s<1/3$. Alternatively, if $A$ is the left-half circle with radius $1/2$ centered at $1$, it again follows from Theorem \ref{3pTHM}(c) that $1/2\notin\supp\lambda_{s,A}$ for $0<s<1/3$. In contrast, as $A$ is moved to the right $R$ units and $R\to\infty$, we get $\supp\lambda_{s,A}^{\infty}=A$; see Lemma \ref{circle}.
\end{eg}


To prove Theorem \ref{3pTHM} we use a special case of the following observation. 
\begin{lem} \label{lem:suff.cond}
Let $0<s<1$. Suppose $A$ is a compact set in the interior of $\Hplus$. Let $\lambda$ denote the unique equilibrium measure on $A$ for $\mathcal{K}_{s}$. If 
\begin{equation}
\label{sufficient.cond}
\mathcal{K}_{s}(z,\cdot) > \int_{B} \mathcal{K}_{s}(\cdot,w^{\prime}) \dd\nu(w^{\prime}) \qquad \text{everywhere on $\supp \lambda$}
\end{equation}
for some subset $B\subset A$ and some probability measure $\nu\in\mathcal{M}(B)$, then $z\notin\supp\lambda$.
\end{lem}
\begin{proof} Using \eqref{sufficient.cond} and the variational inequality \eqref{variat.inequ.a}, we get
\begin{equation*}
\begin{split}
\POTW_{\mathcal{K}_{s}}^{\lambda}(z) &= \int \mathcal{K}_{s}(z,w) \dd\lambda(w) > \int \left[ \int_{B} \mathcal{K}_{s}(w,w^{\prime}) \dd\nu(w^{\prime}) \right] \dd\lambda(w) \\
&= \int_{B}  \POTW_{\mathcal{K}_{s}}^{\lambda}(w^{\prime}) \dd\nu(w^{\prime}) \geq \mathcal{J}_{\mathcal{K}_{s}}[\lambda] \int_{B} \dd\nu(w^{\prime}) = \mathcal{J}_{\mathcal{K}_{s}}[\lambda].
\end{split}
\end{equation*}
But $\POTW_{\mathcal{K}_{s}}^{\lambda}(z)>\mathcal{J}_{\mathcal{K}_{s}}[\lambda]$ implies, by the variational inequality \eqref{variat.inequ.b}, that $z\notin\supp\lambda$. 
\end{proof}

Let $z=x>0$ and set $B=\{z^{\prime},\overline{z^{\prime}}\}$, $z^{\prime}$ in the interior of $\Hplus$, $\im[z^{\prime}]\neq0$, and place the charge $1/2$ at each point in $B$. Then \eqref{sufficient.cond} is equivalent to the property
\begin{equation}
\mathcal{K}_{s}(z,\cdot) > \mathcal{K}_{s}^{*}(\cdot,z^{\prime}) \qquad \text{everywhere on $\supp \lambda$}, \label{super.rel}
\end{equation}
where $\mathcal{K}_{s}^{*}$ denotes the kernel 
\begin{equation}
\mathcal{K}_{s}^{*}(z,w) \DEF \left[ \mathcal{K}_{s}(z,w) + \mathcal{K}_{s}(z,\overline{w}) \right]/2. \label{K.sym}
\end{equation}

\begin{thm}[$3$-point Theorem] \label{3pTHM}
Let $0<s<1$. Let $x>0$ and $z^{\prime}$ be in the interior of $\Hplus$. Let $A$ be a compact subset of $\{w\in\Hplus | \mathcal{K}_{s}(x,w)\geq \mathcal{K}_{s}(x,z^{\prime})\}$ in the interior of $\Hplus$ with $x,z^{\prime},\overline{z^{\prime}}\in A$. 
\begin{enumerate}
\item[(a)] If $\Delta_{s}\DEF\mathcal{K}_{s}(x,z^{\prime}) - \mathcal{K}_{s}^{*}(z^{\prime},z^{\prime})>0$, then $x\notin\supp\lambda_{s,A}$.
\item[(b)] If $z^{\prime}=1+\ii \gamma$, $\gamma>0$, and condition
\begin{equation}
\label{3pCond}
4 \left( \gamma + \sqrt{1+\gamma^{2}} \right) > \left( \sqrt{\left(1+x\right)^{2}+\gamma^{2}} + \sqrt{\left(1-x\right)^{2}+\gamma^{2}} \right)^{2}
\end{equation}
is satisfied, then $\Delta_{s}>0$ (and hence, by (a), $x\notin\supp\lambda_{s,A}$) for $s>0$ sufficiently small.
\item[(c)] If $x=1/2$ and $z^{\prime}=1+\ii/2$, then $\Delta_{s}>0$ (and hence, by (a), $x\notin\supp\lambda_{s,A}$) for all $0<s<1/3$. (The graph of $\Delta_{s}$ is shown in Figure \ref{fig-0}.)
\end{enumerate}
\end{thm}

\begin{proof}[Proof of Theorem \ref{3pTHM}] 
We show first (a). The function $\mathcal{K}_{s}(z,\cdot)$ has a unique maximum at $z$ in $\Hplus$ (Lemma \ref{lem:K.s.properties}(\ref{Prop.5})). So 
\begin{equation} \label{kernel.inequ}
\mathcal{K}_{s}(x,w) - \mathcal{K}_{s}^{*}(w,z^{\prime}) \geq \mathcal{K}_{s}(x,w) - \mathcal{K}_{s}^{*}(z^{\prime},z^{\prime}) \geq \mathcal{K}_{s}(x,z^{\prime}) - \mathcal{K}_{s}^{*}(z^{\prime},z^{\prime}).
\end{equation}
The first inequality holds in $\Hplus$. The last one holds on $\{w\in\Hplus | \mathcal{K}_{s}(x,w)\geq \mathcal{K}_{s}(x,z^{\prime})\}$. Now, let $A$ be a compact subset of $\{w\in\Hplus | \mathcal{K}_{s}(x,w)\geq \mathcal{K}_{s}(x,z^{\prime})\}$ in the interior of $\Hplus$ with $x,z^{\prime},\overline{z^{\prime}}\in A$. Then 
\begin{equation*}
W_{s}^{\nu}(x)\geq \left[ W_{s}^{\nu}(z^{\prime}) + W_{s}^{\nu}(\overline{z^{\prime}}) \right]/2 +  \mathcal{K}_{s}(x,z^{\prime}) - \mathcal{K}_{s}^{*}(z^{\prime},z^{\prime}), \qquad \nu\in\mathcal{M}(A).
\end{equation*}
This follows from \eqref{kernel.inequ} and $\mathcal{K}_{s}(w,z)=\mathcal{K}_{s}(z,w)$. If the difference $\Delta_{s}\DEF\mathcal{K}_{s}(x,z^{\prime}) - \mathcal{K}_{s}^{*}(z^{\prime},z^{\prime})$ is positive, the variational inequality \eqref{variat.inequ.a} implies $W_{s}^{\lambda_{s,A}}(x) > \mathcal{J}_{\mathcal{K}_{s}}[\lambda_{s,A}].$
Therefore, $x\notin\supp\lambda_{s,A}$, by variational inequality \eqref{variat.inequ.b}. This shows (a).

Set $z^{\prime}=\beta+\ii\gamma$ with $\beta,\gamma>0$. Since $\mathcal{K}_{s}(\rho z, \rho w) = \rho^{-s} \mathcal{K}_{s}(z, w)$, $\rho>0$, we may fix one of the variables $x$, $\beta$, or $\gamma$. Let $\beta=1$. From \eqref{K.s.1st}, \eqref{K.sym}, and Lemma \ref{lem:K.s.properties}(\ref{Prop.5}) we get
\begin{equation}
\begin{split}
\Delta_{s} &= \left[ \left( 1 + x \right)^{2} + \gamma^{2} \right]^{-s/2} \Hypergeom{2}{1}{s/2,1/2}{1}{\frac{4x}{\left( 1 + x \right)^{2} + \gamma^{2}}} \\
&\phantom{=\pm}- \frac{1}{2} 2^{-s} \left( 1 + \gamma^{2} \right)^{-s/2} \Hypergeom{2}{1}{s/2,1/2}{1}{\frac{1}{1 + \gamma^{2}}} - \frac{1}{2} 2^{-s} \frac{\gammafcn((1-s)/2)}{\sqrt{\pi} \gammafcn(1-s/2)}. \label{Delta.s:beta.EQ.1}
\end{split}
\end{equation}
We approximate $\Delta_{s}$ by its series expansion at $s=0$. From Lemma \ref{lem:K.s.properties}(\ref{Prop.7}) and \eqref{K.sym} 
\begin{equation*}
\begin{split}
\lim_{s\to0^{+}} \frac{\Delta_{s}}{s} &= \mathcal{K}_{0}(x,1+\ii \gamma) - \mathcal{K}_{0}^{*}(1+\ii \gamma ,1+\ii \gamma) \\
&= \frac{1}{2} \log \frac{4 \left( \gamma + \sqrt{1 + \gamma^{2} } \right) }{\left( \sqrt{\left(1+x\right)^{2}+\gamma^{2}} + \sqrt{\left(1-x\right)^{2}+\gamma^{2}} \right)^{2}} > 0
\end{split}
\end{equation*} 
which implies (b). 

We show that $\Delta_{s}$ (as a function in $s$) is strictly concave on $(0,1)$ if $x=\gamma=1/2$. Using \eqref{K.sym} and integral representation \eqref{K.s:int.repr} we get
\begin{equation*}
\begin{split}
\Delta_{s} 
&\DEF \mathcal{K}_{s}(x,1+\ii\gamma) - \mathcal{K}_{s}^{*}(1+\ii\gamma,1+\ii\gamma) = \frac{1}{\pi} \int_{0}^{1} t^{-1/2} \left(1-t\right)^{-1/2} g(s,t) \dd t,
\end{split}
\end{equation*}
where
\begin{equation*}
g(s,t) \DEF \left[ \left( 1 + x \right)^{2} + \gamma^{2} - 4 x t \right]^{-s/2} - \frac{1}{2} \left[ 4 \left( 1 + \gamma^{2} \right) - 4 t \right]^{-s/2} - \frac{1}{2} \left[ 4 - 4 t \right]^{-s/2}.
\end{equation*}
Negativity of $(\partial/\partial s)^2 g(s,t)$ for all $0\leq t<1$ implies that $\Delta_{s}$ is strictly concave. Let $x=\gamma=1/2$. From $(\partial/\partial s)^2 r^{-s/2}=(1/4)F(s,r)$, $F(s,r) \DEF r^{-s/2}(\log r)^{2}$, we get
\begin{equation*}
\begin{split}
4\left(\frac{\partial}{\partial s}\right)^{2} g(s,t) = F(s,5/2-2t) - (1/2) F(s,5-4t) - (1/2) F(s,4-4t).
\end{split}
\end{equation*}
Negativity of the right-hand side above is equivalent with
\begin{equation*}
\begin{split}
2^{1+s/2} &< \left( \frac{5/2-2t}{5/2-2t} \right)^{s/2} \frac{\left[ \log\left( 5-4t \right) \right]^{2}}{\left[ \log\left( 5/2-2t \right) \right]^{2}} + \left( \frac{5/2-2t}{2-2t} \right)^{s/2} \frac{\left[ \log\left( 4-4t \right) \right]^{2}}{\left[ \log\left( 5/2-2t \right) \right]^{2}}.
\end{split}
\end{equation*}
For growing $s$ the left-hand side of the last relation is increasing while the right-hand side above is decreasing. Thus
\begin{equation} \label{aux.suff.cond}
2^{1+1/2} < \frac{\left[ \log\left( 5-4t \right) \right]^{2}}{\left[ \log\left( 5/2-2t \right) \right]^{2}} + \frac{\left[ \log\left( 4-4t \right) \right]^{2}}{\left[ \log\left( 5/2-2t \right) \right]^{2}} \FED h_{1}(t) + h_{2}(t)
\end{equation}
is a sufficient (but not necessary) condition for $(\partial/\partial s)^2 g(s,t)<0$ for all $0\leq t<1$. Elementary calculus shows that $h_{1}(t)>3$ on $[0,3/4)$ and $h_{2}(t)\geq4$ on $[3/4,1)$. Hence the right side of \eqref{aux.suff.cond} is $>3>2^{3/2}$ on $[0,1)$. Consequently, $\Delta_{s}$ is strictly concave for $0<s<1$ (cf. Figure \ref{fig-0}). 

Since $\Delta_{s}$ has a zero at $s=0$ with $\lim_{s\to0^{+}}\Delta_{s}/s=(1/2)\log(\sqrt{5}-1)>0$, $\lim_{s\to1^{-}}\Delta_{s}=-\infty$, and $\Delta_{s}$ is strictly concave on $(0,1)$, the difference $\Delta_{s}$ has exactly one other zero in the interval $[0,1)$ denoted by $s_{1}$. By our reasoning $\Delta_{s}>0$ if and only if $0<s<s_{1}$. A numerical solver gives $s_{1}\approx0.341107\dots$. 
Numerical computation shows that $\Delta_{1/3}\approx0.0011>0$. This can be rigorously justified by assistance of Mathematica and use of exact arithmetic. 
\end{proof}

By Theorem \ref{3pTHM}(a), the positivity of $\Delta_{s}$ implies $x\not\in\supp\lambda_{s,A}$. By \eqref{Delta.s:beta.EQ.1}, $\Delta_{s}$ depends on three parameters $x$, $\gamma$, and $s$. See Figure \ref{fig-0} for a plot of the level surface $\Delta_{s}=0$. This $0$-level surface is the boundary of the set of admissible configurations $(x,1/\gamma,s)$ using a three point scheme $z=x$, $z^{\prime}=1+\ii\gamma$, and $\overline{z^{\prime}}$. From Figure \ref{fig-0} we get numerical evidence that the maximum $s$ possible for a three point approach is about $0.38$. 
\begin{figure}[ht]
\begin{center}
\end{center}
\caption{$0$-level set of $\Delta_{s}$ in \eqref{Delta.s:beta.EQ.1} cut off at $1/\gamma=4$ and $\Delta_{s}$ for $x=1/2$, $\gamma=1/2$. \label{fig-0}}
\end{figure}

\section{Kernel $\mathcal{K}_{s}$ in the limit $R\to\infty$}
\label{K.s:R.to.infty}

We want to study the behavior of $\mathcal{K}_{s}(R+z,R+w)$ as $R$ becomes large.
\begin{lem} \label{asympt.expansion}
Let $0<s<2$, $s\neq1$, and $z,w\in\Hplus$. Then
\begin{equation}
\begin{split} \label{eq:eq:large.R.kernel.asympt.exp}
&\mathcal{K}_{s}(R+z,R+w) 
= \mathcal{I}_{s}\left(\Sp^{1};\frac{\dd\phi}{2\pi}\right) R^{-s} - \frac{s}{1-s} \frac{\gammafcn((1+s)/2)}{\sqrt{\pi} \gammafcn(1+s/2)} \frac{\left|z-w\right|^{1-s}}{2R} \\
&\phantom{=\pm}- 2^{-s} \frac{s\gammafcn((1-s)/2)}{\sqrt{\pi}\,\gammafcn(1-s/2)} \frac{\re\left[z-w_{*}\right]}{2R} R^{-s} + \mathcal{O}\left(\frac{s}{R^{2}}\right), \qquad R\to\infty,
\end{split}
\end{equation}
where $\mathcal{I}_{s}(\Sp^{1};\dd\phi/(2\pi))$ is given in \eqref{I.s}.
\end{lem}
\begin{rmk}
In the case $1<s<2$ the second term in \eqref{eq:eq:large.R.kernel.asympt.exp} becomes the dominant term. In the special case $s=1$ the following expansion can be shown:
\begin{equation}
\begin{split}
&2R \, \mathcal{K}_1(R+z,R+w) = \frac{6\log2}{\pi} + \frac{2}{\pi} \log R \\
&\phantom{equals}- \frac{2}{\pi} \log|z-w| \left[ 1- \frac{\re[z-w_*]}{2R} + \mathcal{O}(R^{-2}) \right]+ \mathcal{O}( \frac{\log R}{R} ), \qquad R\to\infty.
\end{split}
\end{equation}
\end{rmk}

\begin{proof}[Proof of Lemma \ref{asympt.expansion}]
Let $0<s<2$, $s\neq1$. Using \cite[15.3.6]{abramowitz/stegun:70}, we get a representation of \eqref{K.s.1st},
\begin{equation*}
\begin{split} 
&\mathcal{K}_{s}(R+z,R+w) \\
&\phantom{=}= \frac{\gammafcn((1-s)/2)}{\sqrt{\pi}\,\gammafcn(1-s/2)} \left|2R+z-w_{*}\right|^{-s} \Hypergeom{2}{1}{s/2,1/2}{(1+s)/2}{\frac{\left|z-w\right|^{2}}{\left|2R+z-w_{*}\right|^{2}}} \\
&\phantom{=\pm}- \frac{2}{1-s} \, \frac{\gammafcn((1+s)/2)}{\sqrt{\pi}\,\gammafcn(s/2)} \frac{\left|z-w\right|^{1-s}}{\left|2R+z-w_{*}\right|} \Hypergeom{2}{1}{1-s/2,1/2}{1+(1-s)/2}{\frac{\left|z-w\right|^{2}}{\left|2R+z-w_{*}\right|^{2}}},
\end{split}
\end{equation*}
with convergent series expansions of both hypergeometric functions. The first one is of the form $1+\mathcal{O}(s R^{-2})$, the second one is of the form $1+\mathcal{O}(R^{-2})$. 
Since
\begin{equation*}
\left|1+\frac{z-w_{*}}{2R}\right|^{-s} = 1 - \frac{s}{2} \, \frac{\re[z-w_{*}]}{R} + \mathcal{O}\left(\frac{s}{R^{2}}\right), \qquad R\to\infty,
\end{equation*}
we get
\begin{equation*}
\begin{split} 
&\mathcal{K}_{s}(R+z,R+w) \\
&\phantom{=}=  2^{-s} \frac{\gammafcn((1-s)/2)}{\sqrt{\pi}\,\gammafcn(1-s/2)} \, R^{-s} \left[ 1 - \frac{s}{2} \, \frac{\re[z-w_{*}]}{R} + \mathcal{O}\left(\frac{s}{R^{2}}\right)\right] \left[ 1 + \mathcal{O}\left(\frac{s}{R^{2}}\right) \right] \\
&\phantom{=\pm}- \frac{1}{1-s} \, \frac{\gammafcn((1+s)/2)}{\sqrt{\pi}\,\gammafcn(s/2)} \, \left|z-w\right|^{1-s} \, R^{-1} \left[ 1 + \mathcal{O}\left(\frac{1}{R}\right)\right] \left[ 1 + \mathcal{O}\left(\frac{1}{R^{2}}\right) \right].
\end{split}
\end{equation*}
We reorder the terms with respect to powers of $R$ and obtain \eqref{eq:eq:large.R.kernel.asympt.exp}.
\end{proof}

It is convenient to define the following kernels
\begin{align}
\mathcal{K}_{s}^{(R)}(z,w) &\DEF 2R \left[ \mathcal{K}_{s}(R+z,R+w) - \mathcal{I}_{s}(\Sp^{1};\dd\phi/(2\pi)) \, R^{-s} \right], \quad 0<s<1, \\
\mathcal{K}_{s}^{(R)}(z,w) &\DEF 2R \, \mathcal{K}_{s}(R+z,R+w), \quad s>1, \label{KRdef2}
\intertext{and}
\mathcal{K}_{s}^{(\infty)}(z,w) &\DEF - \frac{2}{1-s} \, \frac{\gammafcn((1+s)/2)}{\sqrt{\pi}\,\gammafcn(s/2)} \, \left|z-w\right|^{1-s} = \frac{\gammafcn((s-1)/2)}{\sqrt{\pi}\,\gammafcn(s/2)} \, \left|z-w\right|^{1-s}.
\end{align}
Then, by \eqref{eq:eq:large.R.kernel.asympt.exp},
\begin{equation}
\lim_{R\to\infty} \mathcal{K}_{s}^{(R)}(z,w) = \mathcal{K}_{s}^{(\infty)}(z,w), \qquad 0<s<1, \label{K.R.limit}
\end{equation}
and, from \eqref{K.s.3rd} and \cite[15.1.20]{abramowitz/stegun:70}, it follows 
\begin{equation}
\lim_{R\to\infty} \left|z-w\right|^{s-1}\mathcal{K}_{s}^{(R)}(z,w) = \left|z-w\right|^{s-1} \mathcal{K}_{s}^{(\infty)}(z,w), \qquad s>1, \label{K.R.limit.s.GT.1}
\end{equation}
where in both cases the convergence is uniform on compact subsets of $\Hplus\times\Hplus$. If $s<\dim\Gamma(A)$, we let $\mathcal{J}_{K_{s}^{(R)}}[\nu]$ and $\mathcal{J}_{K_{s}^{(\infty)}}[\nu]$ denote the associated energies of the compactly supported measure $\nu\in\mathcal{M}(\Hplus)$. From the definition of the kernel $\mathcal{K}_{s}^{(R)}$ we see that the equilibrium measure $\lambda_{s,A}^{R}$ on the compact set $A\subset\Hplus$ for the kernel $\mathcal{K}_{s}^{(R)}$ is equal to the equilibrium measure $\lambda_{s,R+A}$ on $R+A$ for the kernel $\mathcal{K}_{s}$ in the following sense: $\lambda_{s,A}^{R}(B)=\lambda_{A+R}(R+B)$ for a Borel set $B\subset\Hplus$.
\begin{rmk} \label{rmk:reorder.limits}
The asymptotics \eqref{eq:eq:large.R.kernel.asympt.exp} holds uniformly in $0\leq s\leq s_{0}<1$. So
\begin{equation*}
\lim_{s\to0^{+}} \mathcal{K}_{s}^{(R)}(z,w) / s = \mathcal{K}_{0}^{(\infty)}(z,w) + \mathcal{O}\left(1/R \right), \qquad R\to\infty.
\end{equation*}
The expression $\mathcal{K}_{0}^{(\infty)}(z,w)\DEF -\re[z-w_{*}] - |z-w|$ is the $\infty$-kernel for the logarithmic case introduced in  \cite{hardin/saff/stahl:2007}. However, reversing the order of limit processes, we get
\begin{equation*}
\lim_{R\to\infty} \mathcal{K}_{s}^{(R)}(z,w) / s = \mathcal{K}_{s}^{(\infty)}(z,w) / s.
\end{equation*}
Now, in the limit $s\to0^{+}$, the right-hand side above tends to $-|z-w|$.
\end{rmk}

\section{The energy problem for the kernel $\mathcal{K}_{s}^{(\infty)}$}
\label{sec:K.s.infty}

\subsection{The case $0<s<1$} 
The kernel
\begin{equation*}
\mathcal{K}_{s}^{(\infty)}(z,w) = - \frac{2}{1-s} \frac{\gammafcn((1+s)/2)}{\sqrt{\pi} \gammafcn(s/2)} \left|z-w\right|^{1-s}, \qquad 0<s<1,
\end{equation*}
falls into a class of kernels studied by Bj{\"o}rck \cite{bjoerck:1956}. From his results we infer that to every compact set $A\subset\Hplus$ and every $0<s<1$ there exists a unique equilibrium measure $\lambda_{s,A}^{\infty}$ supported on the outer boundary of $A$. (``Outer boundary'' is justified by the strict superharmonicity of the infinity kernel everywhere in $\Cset$.) Let $\POTW_{\mathcal{K}_{s}^{(\infty)}}^{\mu}$ denote the potential for a measure $\mu\in\mathcal{M}(A)$ and for the kernel $\mathcal{K}_{s}^{(\infty)}$:
\begin{equation*}
\POTW_{\mathcal{K}_{s}^{(\infty)}}^{\mu}(z) \DEF \int_{A} \mathcal{K}_{s}^{(\infty)}(z,w) \, \dd\mu(w), \qquad z\in\Hplus.
\end{equation*}
Then $\POTW_{\mathcal{K}_{s}^{(\infty)}}^{\mu}$ is continuous on $\Hplus$ and from results in \cite{bjoerck:1956} there follows that $\POTW_{\mathcal{K}_{s}^{(\infty)}}^{\lambda_{s,A}^{\infty}}\geq\mathcal{J}_{\mathcal{K}_{s}^{(\infty)}}[\lambda_{s,A}^{\infty}]$ on $A$ and equality holds on $\supp\lambda_{s,A}^{\infty}$. We note, that $\lambda_{s,A}^{R}$ converges weak-star to $\lambda_{s,A}^{\infty}$ as $R\to\infty$. This follows from the weak-star compactness of $\mathcal{M}(A)$, relation \eqref{K.R.limit}, and the uniqueness of the equilibrium measure $\lambda_{s,A}^{\infty}$.

Suppose the curve $\gamma:[a,b]\to\Hplus$ covers the outer boundary $S$ of $A$. Set $r_{w}=|\gamma(t)-w|$. Assuming $\gamma$ is twice differentiable at $t$ we have 
\begin{equation}
\frac{\dd^{2}}{\dd t^{2}} \mathcal{K}_{s}^{(\infty)}(\gamma(t),w) = 2 \frac{\gammafcn((1+s)/2)}{\sqrt{\pi} \gammafcn(s/2)} \left[ s \left( r_{w}^{\prime} \right)^{2} - r_{w} r_{w}^{\prime\prime} \right] r_{w}^{-s-1} . \label{conv.cond}
\end{equation}
Then for fixed $w$, we have that $\mathcal{K}_{s}^{(\infty)}(\gamma(t),w)$ is strictly convex on any interval where $s ( r_{w}^{\prime} )^{2} > r_{w} r_{w}^{\prime\prime}$. A sufficient condition would be $r_{w}^{\prime\prime}<0$. 

In the following we give examples of compact sets $A\subset\Hplus$ such that the support of the equilibrium measure on $A$ is given by the outer  boundary of $A$.

\begin{lem} \label{line-segment}
Let $A$ be a compact subset of the line-segment $[z^{\prime},z^{\prime\prime}]$ in the interior of $\Hplus$, $z^{\prime\prime}-z^{\prime}=2r e^{\ii\phi}$, $r>0$, $0\leq\phi<\pi$. Then $\supp\lambda_{s,A}^{\infty}=A$ for all $0<s<1$. In particular, if $A=[z^{\prime},z^{\prime\prime}]$, then
\begin{equation} \label{lambda.line-segment}
\dd \lambda_{s,A}^{\infty}(w) = \frac{\gammafcn((1+s)/2)}{\sqrt{\pi}\gammafcn(s/2)} r^{1-s} \left( r^{2} - T^{2} \right)^{s/2-1} \dd T,
\end{equation}
where $w=(z^{\prime}+z^{\prime\prime})/2 + T  e^{\ii\phi}$, $|T|\leq r$. 
\end{lem}
\begin{proof}
W.l.o.g. consider the parametrization $\gamma(t)= t e^{\ii\phi}$, $|t|\leq r$. Then
\begin{equation*}
\frac{\dd^{2}}{\dd t^{2}} \mathcal{K}_{s}^{(\infty)}(\gamma(t),\gamma(T)) = 2 \frac{s \gammafcn((1+s)/2)}{\sqrt{\pi} \gammafcn(s/2)}  \left|t-T\right|^{-1-s}>0, \qquad 0<s<1.
\end{equation*}
By Theorem \ref{thm:convexity} there exists an interval $I\subset[-1,1]$ such that $\supp\lambda_{s,A}^{\infty}=\gamma(I)\cap A$. Since the kernel $\mathcal{K}_{s}^{(\infty)}(\gamma(t),\gamma(T))$ decreases as $|t-T|$ grows, there follows that the equilibrium potential is strictly less than $\mathcal{J}_{\mathcal{K}_{s}^{(\infty)}}[\lambda_{s,A}^{\infty}]$ on $(-\infty e^{\ii\phi},\infty e^{\ii\phi}) \setminus \gamma(I)$. But the equilibrium potential is $\geq\mathcal{J}_{\mathcal{K}_{s}^{(\infty)}}[\lambda_{s,A}^{\infty}]$ on $A$. So, $\supp\lambda_{s,A}^{\infty}=A$. Relations \eqref{lambda.line-segment} 
follows from the constancy of the integral
\begin{equation*}
\begin{split}
\int_{-r}^{r} \left| t - T \right|^{1-s} \left( r^{2} - T^{2} \right)^{s/2-1} \dd T = \gammafcn(s/2) \gammafcn(1-s/2)
\end{split}
\end{equation*}
and the fact that the $\mathcal{K}_{s}^{(\infty)}$-potential for this measure \eqref{lambda.line-segment} is strictly decreasing away from the line-segment. We used the auxiliary result Lemma \ref{hilfssatz}.
\end{proof}

\begin{lem}[{\cite[Hilfssatz~I]{polya/szego:1931}}] \label{hilfssatz}
Let $-1<\alpha<1$, $\alpha\neq0$. Then for $-1\leq y\leq 1$:
\begin{equation*}
\int_{-1}^{1} \left(1-x^{2}\right)^{-(1+\alpha)/2} \left| x - y \right|^{\alpha} \dd x = \gammafcn\left(\frac{1-\alpha}{2}\right) \gammafcn\left(\frac{1+\alpha}{2}\right) = \frac{\pi}{\cos(\pi\alpha/2)}.
\end{equation*}
\end{lem}

\begin{lem} \label{circle}
Let the outer boundary $S$ of the compact set $A$ be a subset of a circle $C$ centered at $a>0$ with radius $0<r<a$. Then $\supp\lambda_{s,A}^{\infty}=S$ for every $0<s<1$. In particular, if $S=C$, then $\lambda_{s,A}^{\infty}$ is given by the normalized arc-length measure on $C$ and $\supp\lambda_{s,A}^{\infty}=C$ for all $0<s<1$. \end{lem}

The result $\supp\lambda_{s,A}^{\infty}=S$ for $0<s<1$ differs considerably from the logarithmic case. By \cite[Thm.~4]{hardin/saff/stahl:2007}, one has $\supp\lambda_{0,A}^{\infty}=\{ a + r e^{\ii\phi} \mid |\phi|\leq\theta \}$ for some $\theta\in[0,\pi/3]$.

\begin{proof}
W.l.o.g. consider the parametrization $\gamma(\phi) = r e^{\ii \phi}$, $0\leq\phi\leq2\pi$. Then
\begin{equation*}
\mathcal{K}_{s}^{(\infty)}(\gamma(\phi),\gamma(\phi^{\prime})) = - \frac{2^{2-s}}{1-s} \frac{\gammafcn((1+s)/2)}{\sqrt{\pi} \gammafcn(s/2)}  \left|\sin\frac{\phi-\phi^{\prime}}{2}\right|^{1-s} r^{1-s}.
\end{equation*}
By direct calculation (assisted by Mathematica) 
\begin{equation*}
\frac{\dd^{2}}{\dd \phi^{2}}\mathcal{K}_{s}^{(\infty)}(\gamma(\phi),\gamma(\phi^{\prime})) = 2^{-s-1} \frac{\gammafcn((1+s)/2)}{\sqrt{\pi} \gammafcn(s/2)} \frac{1 + s - \left( 1 - s \right) \cos \left( \phi - \phi^{\prime} \right)}{\left|\sin\left[\left(\phi-\phi^{\prime}\right)/2\right]\right|^{1+s}} r^{1-s} > 0.
\end{equation*}
Since $\gamma$ is a simple closed continuous curve and $\gamma(\phi)=\gamma(\phi+2\pi)$, by Theorem \ref{thm:convexity}, $\supp\lambda_{s,A}^{\infty}=S$. In the case $S=C$, rotational symmetry gives $\dd\lambda_{s,A}^{\infty}=\dd\phi/(2\pi)$. 
\end{proof}

\subsection{The case $s>1$} The kernel
\begin{equation*}
\mathcal{K}_{s}^{(\infty)}(z,w) = \frac{\gammafcn((s-1)/2)}{\sqrt{\pi} \gammafcn(s/2)} \frac{1}{\left|z-w\right|^{s-1}}, \qquad s>1,
\end{equation*}
is (up to a multiplicative constant) the Riesz-$(s-1)$-kernel in the plane $\Rset^2$ which can be identified with $\Cset$. If $1<s<1+\dim A$, then classical potential theory yields that there exists a unique equilibrium measure $\lambda_{s,A}^\infty$ on $A$ with $\supp\lambda_{s,A}^\infty\supset\check{A}$, where $\check{A}$ denotes the set of all points of $A$ each neighborhood of which intersects $A$ in a set of positive Riesz $(s-1)$-capacity. Examples of sets $A$ with $A=\check{A}$ are line-segments, circles, or more generally, any Jordan curve; discs, ``washers''.

\begin{lem}
Let $A$ be a compact subset of $\Cset$ with $\dim A>0$ and $s$ a real number with $1<s<1+\dim A$. Then $\lambda_{s,A}^{R}$ converges weak-star to $\lambda_{s,A}^{\infty}$ as $R\to\infty$. 
\end{lem}

\begin{proof}
To simplify notation we use the abbreviations $\mathcal{K}_R=\mathcal{K}_s^{(R)}$, $\mathcal{K}_\infty=\mathcal{K}_s^{(\infty)}$, $\lambda_R=\lambda_{s,A}^R$, and $\lambda_\infty=\lambda_{s,A}^\infty$. From the definition of $\mathcal{K}_R$ given
in \eqref{KRdef2} and the formula \eqref{K.s.3rd} it follows that 
\begin{equation*}
\mathcal{K}_R(z,w) = \Omega_R(z,w) \mathcal{K}_\infty(z,w), \qquad (z,w)\in \Cset\times \Cset,
\end{equation*}
where
\begin{equation*}
\Omega_R(z,w) \DEF \frac{\Hypergeom{2}{1}{1-s/2,1/2}{1}{1-\frac{\left|z-w\right|^{2}}{\left|2R+z-w_{*}\right|^{2}}}}{\left|1+\frac{z-w_*}{2R}\right|\Hypergeom{2}{1}{1-s/2,1/2}{1}{1}}.
\end{equation*}
We remark that $\Omega_R$ converges uniformly to 1 on compact subsets of $\Cset\times \Cset$ as $R\to\infty$.

Since $\mathcal{M}(A)$ is weak-star-compact, there exists a weak-star cluster point $\lambda^*$ of $\lambda_R$ as $R\to \infty$.  We will show that $\mathcal{J}_{\mathcal{K}_\infty}[\lambda^*]\leq \mathcal{J}_{\mathcal{K}_\infty}[\lambda_\infty]$ from which the Lemma will immediately follow.   Let  $R_k$, $k\geq1$, be a sequence of numbers such that $\lim_{k\to\infty}R_k=\infty$ and  $\lambda_{R_k}\weakstarto\lambda^*\in\mathcal{M}(A)$ as $k\to\infty$. Thus $\Omega_{R_k}(\lambda_{R_k}\times\lambda_{R_k})\weakstarto\lambda^*\times\lambda^*$ as $k\to\infty$ and   we have (see \cite[Lemma~0.1]{landkof:72})
\begin{equation*}
\mathcal{J}_{\mathcal{K_\infty}}[\lambda^*]\le \liminf_{k\to\infty} \mathcal{J}_{\mathcal{K}_{R_k}}[\lambda_{R_k}]\le \liminf_{k\to\infty} \mathcal{J}_{\mathcal{K}_{R_k}}[\lambda_\infty],
\end{equation*}
where the second inequality follows since $\lambda_{R_k}$ minimizes $\mathcal{J}_{\mathcal{K}_{R_k}}$.  Finally,
since $\Omega_R$ converges uniformly to 1 on $A\times A$ as $R\to\infty$, we have $\liminf_{k\to\infty} \mathcal{J}_{\mathcal{K}_{R_k}}[\lambda_\infty]= \mathcal{J}_{\mathcal{K_{\infty}}}[\lambda_\infty]$ which shows that $\mathcal{J}_{\mathcal{K_\infty}}[\lambda^*]\le \mathcal{J}_{\mathcal{K_\infty}}[\lambda_\infty]$.  Since  $\mathcal{J}_{\mathcal{K_\infty}}$ has a unique minimizer, it follows that $\lambda_\infty$ is the only weak-star cluster point of $\lambda_R$ as $R\to \infty$.
\end{proof}

In the hyper-singular case $s>1+\dim A=\dim\Gamma(A)$, both energy integrals $\mathcal{J}_{\mathcal{K}_s^{(R)}}[\nu]$ and $\mathcal{J}_{\mathcal{K}_s^{(\infty)}}[\nu]$ are infinite for every $\nu\in\mathcal{M}(A)$. In 
Section~\ref{hypersingular} and \ref{hypersing.K.infty} we consider the limit distribution of minimal $\mathcal{K}_s$-energy and  $\mathcal{K}_s^{(\infty)}$-energy $N$-point systems as $N\to\infty$ for the hyper-singular case and for sufficiently ``nice" sets $A$ (namely $d$-rectifiable sets).

\section{Discrete Minimum Energy problems on $A\subset\Hplus$}
\label{sec:num.experiments}

In this section we discuss the discrete Riesz $s$-energy problem on $\Gamma(A)\subset \Rset^3$ as well as the discrete $\mathcal{K}$-energy problem  on $A\subset\Hplus$ for the kernel $\mathcal{K}=\mathcal{K}_{s}$, $\mathcal{K}=\mathcal{K}_{s}^{(R)}$, and $\mathcal{K}=\mathcal{K}_{s}^{(\infty)}$.  The  {\em $N$-point Riesz $s$-energy of $\Gamma(A)$} is defined as
\begin{equation*}
\mathcal{E}_s(\Gamma(A),N) \DEF \min E_s(X_{N}),
\end{equation*}
where the minimum is taken over all $N$-point configurations $X_N\subset \Gamma(A)$ and $E_s(X_{N})$ is defined as in \eqref{Esdef}.  We let $X_N^*=X_{N,s}^*$ denote an $N$-point configuration in $\Gamma(A)$ attaining this minimum.  

Similarly, for an $N$-point configuration $Z_N=\{z_1,\dots,z_N\}\subset A$, let 
\begin{equation*}
E_{\mathcal{K}}(Z_{N}) \DEF \sum_{j\neq k} \mathcal{K}(z_{j},z_{k})
\end{equation*}
and let the {\em $N$-point $\mathcal{K}$-energy of $A$} be defined as
\begin{equation*}
\mathcal{E}_{\mathcal{K}}(A,N) \DEF \min E_{\mathcal{K}}(Z_{N}),
\end{equation*}
over all $N$-point configurations $Z_N\subset A$. This minimum is attained at a minimal $\mathcal{K}$-energy $N$-point system $Z_N^*=\{z_1^*,\dots,z_N^*\}$, that is $\mathcal{E}_{\mathcal{K}}(A,N) = E_{\mathcal{K}}(Z_{N}^*)$. Finally, let $\lambda(Z_N^*) \DEF (1/N) \sum_{k=1}^N \delta_{z_k^*}$. We are interested in the weak-star convergence of $\lambda(Z_N^*)$ and in the asymptotic growth of $\mathcal{E}_{\mathcal{K}}(A,N)$ as $N\to\infty$.

\subsection{The potential theory case} \label{pottheory}
For $0<s<\dim\Gamma(A)$, there is a unique equilibrium measure $\lambda_{\mathcal{K},A}$ minimizing the $\mathcal{K}$-energy
\begin{equation*}
\mathcal{J}_{\mathcal{K}}[\lambda] = \int \mathcal{K}(z,w) \dd\lambda(z) \dd\lambda(w)
\end{equation*}
over measures $\lambda\in\mathcal{M}(A)$. (See Proposition \ref{prop.K.s} for the case $\mathcal{K}=\mathcal{K}_s$ and Bj{\"o}rck \cite{bjoerck:1956} for the case $\mathcal{K}=\mathcal{K}_s^{(\infty)}$.)

\begin{prop}
Suppose $A$ is an infinite compact subset in the interior of $\Hplus$. Let $\mathcal{K}=\mathcal{K}_s$, $\mathcal{K}=\mathcal{K}_s^{(R)}$, or $\mathcal{K}=\mathcal{K}_s^{(\infty)}$ and $0<s<\dim\Gamma(A)$. For $N\geq2$, let $Z_N^*$ be a minimal $\mathcal{K}$-energy configuration of $N$ points $\{z_1^*,\dots,z_N^*\}\subset A$. Then $\lambda(Z_N^*)$ converges weak-star to the equilibrium measure $\lambda_{\mathcal{K},A}$ on $A$ as $N\to\infty$ and 
\begin{equation}\label{transfinite}
\lim_{N\to\infty} \frac{\mathcal{E}_{\mathcal{K}}(A,N)}{N^2} = \mathcal{J}_{\mathcal{K}}[\lambda_{\mathcal{K},A}].
\end{equation}
\end{prop}

\begin{proof}
The proof follows using standard arguments as in  \cite{polya/szego:1931},  \cite[pp.~160--162]{landkof:72} and \cite{farkas/nagy:_unpubl}.   The essential ingredients of the proof are the boundedness of the sequence \eqref{transfinite} and existence and uniqueness of the equilibrium measure $\lambda_{\mathcal{K}, A}$.
\end{proof}

In Figure \ref{num_exp1} we show minimal $\mathcal{K}_s$-energy configurations for $N=32$ points restricted to a Cassinian oval for various values of $s$ with $0<s<1$. A somewhat surprising result of these numerical experiments is that for fixed $N$ and for $s$ close to $1^-$ a rather large part of the Cassinian oval is free of points. (Note, that in the case $s=1$, the support of the equilibrium measure $ \lambda_{s,A}$ is   $A$.) In Figure \ref{num_exp2} we show minimal $\mathcal{K}_s^\infty$-energy configurations for $N=32$ points restricted to a Cassinian oval for various values of $s$ with $0<s<1$. 
\begin{figure}[h]
\begin{center}
\caption{\label{num_exp1} Minimum $\mathcal{K}_s$-energy configurations ($N=32$ points) 
for $s=0$, $0.5$, $0.75$, and $1$.}
\end{center}
\end{figure}
\begin{figure}[h]
\begin{center}
\caption{\label{num_exp2} Minimum $\mathcal{K}_s^\infty$-energy configurations ($N=32$ points) 
for $s=0$, $0.5$, $0.75$, and $1$.}
\end{center}
\end{figure}

Numerical experiments for a circle $C$ centered at $a>0$ with radius $r$ ($0<r<a$) suggest that for a fixed $s$, say $s=1/4$, the equilibrium measure $\lambda_{s,C}$ is concentrated on a proper subset of $C$. (For what we can prove, see Example \ref{eg:circle}.) However, the equilibrium measure $\lambda_{s,C}^R$ associated with the translate $R+C$ converges weak-star to $\lambda_{s,C}^\infty$ as $R\to\infty$ and we can show (see Lemma \ref{circle}) that $\lambda_{s,C}^\infty$ is the uniform measure on the circle $C$ for all $0<s<1$. This phenomenon that the support of $\lambda_{s,C}^R$ seems to spread out as $R\to\infty$ is illustrated by considering discrete minimal $\mathcal{K}_s$-energy points on the translate $R+C$ for varying values of $R$. In Figure \ref{num_exp3} we show minimal $\mathcal{K}_s$-energy configurations for $N=40$ points restricted to translates $R+C$ of the unit circle $C$ centered at $a=1$, where $R=10^{k/2}$ ($k=0,1,\dots,5$).
\begin{figure}[h]
\begin{center}
\caption{\label{num_exp3} Minimum $\mathcal{K}_s$-energy configurations ($N=40$ points) on translates $R+C$ of the unit circle $C$ centered at $1$ for $R=10^{k/2}$ ($k=0,1,\dots,5$).}
\end{center}
\end{figure}

\subsection{The hypersingular case for $d$-rectifiable sets} \label{hypersingular} Suppose $s\geq\dim\Gamma(A)$. In this subsection we require that $A$ be a $d$-rectifiable subset of the interior of $\Hplus$ ($d=1,2$). Recall that a set $K\subset\Rset^p$ is {\em $d$-rectifiable}, $d\leq p$, if it is the image of a bounded set $B$ in $\Rset^d$ with respect to a Lipschitz mapping, that is a mapping $\phi:B\to\Rset^p$ that satisfies for some positive constant $c$
\begin{equation}
\left| \phi(x) - \phi(y) \right| \leq c \left| x - y \right| \qquad \text{for all $x,y\in B$.}
\end{equation}
Note that every compact subset of $\Rset^2$ is $2$-rectifiable. Also note that if $A$ is a $d$-rectifiable set in $\Hplus$, then $\Gamma(A)$ is a $(d+1)$-rectifiable set in $\Rset^3$, $d=1,2$. In order to avoid complications, we require that $A$ is in the interior of $\Hplus$. In this case $\dim\Gamma(A)=1+\dim A$.

Using the properties of $\mathcal{K}_s$ as given in Lemma \ref{lem:K.s.properties} it follows that $\mathcal{K}_s(z,w)=\Omega(z,w) |z-w|^{1-s}$, where $\Omega:A\times A\to\Rset$ is continuous and positive. In the terminology of \cite{borodachov/hardin/saff:weighted}   then $\Omega$ is a {\em CPD weight function} on $A$ (see \cite{borodachov/hardin/saff:weighted} for the  general definition of CPD weight function).  Also, note that 
\begin{equation*}
\Omega(w,w) = \frac{\gammafcn((s-1)/2)}{\sqrt{\pi}\gammafcn(s/2)} \frac{1}{\left| w - w_* \right|}.
\end{equation*}
If $A$ is a compact set in $\Rset^p$ and $\Omega$ is a CPD-weight function on $A\times A$, then for $s\geq d$ one can define the {\em weighted Hausdorff measure} $\mathcal{H}_d^{s,\Omega}$ on Borel sets $B\subset \Hplus$ by
\begin{equation}
\mathcal{H}_d^{s,\Omega}(B) \DEF \int_{B\cap A} \left[ \Omega(w,w) \right]^{-d/s} \dd\mathcal{H}_d(w).
\end{equation}
Then the following result is a corollary of Theorem 2 in \cite{borodachov/hardin/saff:weighted}.

\begin{prop} \label{prop:hyper-singular}
Let $d=1$ or $2$ and suppose $A$ is a compact $d$-rectifiable set contained in the interior of $\Hplus$ with positive $d$-dimensional Hausdorff measure $\mathcal{H}_d(A)$. Let $s>\dim\Gamma(A)$. For $N\geq2$, let $Z_N^*$ be a minimal $\mathcal{K}_s$-energy configuration of $N$ points $\{z_1^{s,N},\dots,z_N^{s,N}\}\subset A$. Then the sequence $Z_N^*$, $N\geq2$, is asymptotically uniformly distributed with respect to $\mathcal{H}_d^{s-1,\Omega}$; that is, 
\begin{equation}\label{hypersing1}
\frac{1}{N}\sum_{k=1}^N\delta_{z_k^{s,N}} \weakstarto \frac{\mathcal{H}_d^{s-1,\Omega}}{\mathcal{H}_d^{s-1,\Omega}(A)}, \qquad \text{as $N\to\infty$.}
\end{equation}
Moreover, the minimal $N$-point $\mathcal{K}_s$-energy satisfies
\begin{equation} \label{weighted.energy}
\lim_{N\to\infty} \frac{\mathcal{E}_{\mathcal{K}_s}(A,N)}{N^{1+(s-1)/d}} = \frac{C_{s-1,d}}{\left[ \mathcal{H}_d^{s-1,\Omega}(A), \right]^{(s-1)/d}},
\end{equation}
where $C_{s-1,d}$ is a positive constant which does not depend on $A$ and $N$.
\end{prop}

\begin{rmk} \label{rmk:C.s.d}
The constant $C_{s-1,d}$ is exactly the same constant which appears in the analogue of \eqref{weighted.energy} for the non-weighted case, that is for $\Omega(z,w)=1$ for all $z,w\in A$. It can be represented using the Riesz $s$-energy for the unit cube in $\Rset^d$ via 
\begin{equation} \label{Csd}
C_{s,d} = \lim_{N\to\infty}\frac{\mathcal{E}_s([0,1]^d,N)}{N^{1+s/d}}, \qquad s>d.
\end{equation}
It was shown in \cite{martinez-finkelshtein/maymeskul/rakhmanov/saff} that $C_{s,1}=2\zetafcn(s)$, where $\zeta(s)$ is the classical Riemann zeta function. However, for other values of $d$, the constant $C_{s,d}$ is as yet unknown. 
\end{rmk}

In the boundary case $s=\dim\Gamma(A)$ and for a $1$-rectifiable set $A$ an additional regularity condition is needed to prove a result analogous to Proposition \ref{prop:hyper-singular}. The following result is a corollary of Theorem~3 in \cite{borodachov/hardin/saff:weighted}.
\begin{prop} \label{prop:hyper-singular.boundary}
Let $d=1$ or $2$ and suppose $A$ is a compact $d$-rectifiable set contained in the interior of $\Hplus$ with positive $d$-dimensional Hausdorff measure $\mathcal{H}_d(A)$. If $d=1$, we further require that $A$ is a subset of a $C^1$-curve. Let $s=1+d$. For $N\geq2$, let $Z_N^*$ be a minimal $\mathcal{K}_{d+1}$-energy configuration of $N$ points $\{z_1^{d+1,N},\dots,z_N^{d+1,N}\}\subset A$. Then the sequence $Z_N^*$, $N\geq2$, is asymptotically uniformly distributed with respect to $\mathcal{H}_d^{d,\Omega}$; that is, 
\begin{equation*}
\frac{1}{N}\sum_{k=1}^N \delta_{z_k^{d+1,N}} \weakstarto \frac{\mathcal{H}_d^{d,\Omega}}{\mathcal{H}_d^{d,\Omega}(A)}, \qquad \text{as $N\to\infty$.}
\end{equation*}
Moreover, the minimal $N$-point $\mathcal{K}_s$-energy satisfies
\begin{equation} \label{weighted.energy.boundary.case}
\lim_{N\to\infty} \frac{\mathcal{E}_{\mathcal{K}_{d+1}}(A,N)}{N^2\log N} = \frac{\beta_d}{\mathcal{H}_d^{d,\Omega}(A)},
\end{equation}
where $\beta_d=\pi^{d/2}/\gammafcn(1+d/2)$ is the volume of the unit ball in $\Rset^d$.
\end{prop}

\begin{rmk}
It is a consequence of Theorem~2 and 3 in \cite{borodachov/hardin/saff:weighted}   that  minimal Riesz $s$-energy point   configurations $X_N^*\subset \Gamma(A)$   are asymptotically uniformly distributed with respect to $\mathcal{H}_{ d+1} $ 
restricted to $\Gamma(A)$ in the hypersingular case $s\ge d+1$.   For $z\in \Hplus$, let $\hat \delta_z$ denote the 
rotationally symmetric probability measure supported on $\Gamma(\{z\})$.  When $s<d+1$ we have that both 
$(1/N)\sum_{x\in X_N^*}\delta_x$ and $(1/N)\sum_{z\in Z_N^*}\hat \delta_z$ converge weak-star to $\mu_{s,\Gamma(A)}$. 
However, for $s>d+1$,  the discrete probability measure $(1/N)\sum_{x\in X_N^*}\delta_x$  converge weak-star to 
$\mathcal{H}_{ d+1} $ (normalized and restricted to $\Gamma(A)$), while Proposition~\ref{prop:hyper-singular} implies that $(1/N)\sum_{z\in Z_N^*}\hat \delta_z$ 
converges to a measure that depends on $s$.  In the boundary case $s=d+1$, the latter limit distributions are equal (cf. Proposition~\ref{prop:hyper-singular.boundary}). 
\end{rmk}

In the following we consider two examples: a line-segment in general position and a circle centered on the real axis.
\begin{eg}
Let $A$ be the line segment with parametrization $\gamma(t)=R+ t e^{\ii\phi}$, $|t|\leq 1$, where $R>\cos\phi$ and $0\leq\phi<\pi$ fixed. Then $A$ is a $1$-rectifiable set and the weighted Hausdorff measure $\mathcal{H}_1^{s-1,\Omega}$ can be explicitly calculated. Indeed, since $\dd\mathcal{H}_1(t)= \dd t$, we get for $s>2$
\begin{equation} \label{line-seg-density}
\dd \frac{\mathcal{H}_d^{s-1,\Omega}}{\mathcal{H}_d^{s-1,\Omega}(A)} (t) = \frac{\left( R + t \cos\phi \right)^{1/(s-1)} \dd t}{\dfrac{s-1}{s} \dfrac{\left(R+\cos\phi\right)^{s/(s-1)}-\left(R-\cos\phi\right)^{s/(s-1)}}{\cos\phi}}, \quad |t|\leq1.
\end{equation}
Note, in the case of the vertical line-segment (that is $\phi=\pi/2$), the last expression reduces to
\begin{equation} 
\dd \frac{\mathcal{H}_d^{s-1,\Omega}}{\mathcal{H}_d^{s-1,\Omega}(A)} (t) = \frac{1}{2} \dd t, \quad |t|\leq1; s>2.
\end{equation}
In Figure \ref{num_exp4} we show minimal $\mathcal{K}_2$ and $\mathcal{K}_4$-energy configurations for $N=40$ points restricted to the line-segment with $R=3/2$ and $\phi=\pi/4$. 
\end{eg}
\begin{figure}[h]
\centering
\caption{\label{num_exp4} Minimum $\mathcal{K}_s$-energy configurations ($N=40$ points) on a line-segment for $s=2$ (left) and $s=4$ (right).}
\end{figure}

\begin{eg}
Let $A$ be the unit circle centered at $R>1$. Then $A$ is a $1$-rectifiable set and the weighted Hausdorff measure $\mathcal{H}_1^{s-1,\Omega}$ can be explicitly calculated. Since $\dd\mathcal{H}_1(\phi)=\dd\phi$, one has for $s>2$
\begin{equation} \label{circle-density}
\dd \frac{\mathcal{H}_d^{s-1,\Omega}}{\mathcal{H}_d^{s-1,\Omega}(A)} (\phi) = \frac{1}{2\pi} \frac{\left( \dfrac{R+\cos\phi}{R+1} \right)^{1/(s-1)} \dd\phi}{\Hypergeom{2}{1}{-1/(s-1),1/2}{1}{\dfrac{2}{1+R}}}, \quad -\pi\leq\phi\leq\pi.
\end{equation}
In Figure \ref{num_exp5} we show minimal $\mathcal{K}_2$ and  $\mathcal{K}_4$-energy configurations for $N=40$ points restricted to the unit circle centered at $R=3/2$. \end{eg}
\begin{figure}[h]
\centering
caption{\label{num_exp5} Minimum $\mathcal{K}_s$-energy configurations ($N=40$ points) on the unit circle centered at $3/2$ for $s=2$ (left) and $s=4$ (right).}
\end{figure}

\begin{rmk}
Results similar to Proposition \ref{prop:hyper-singular} and Proposition \ref{prop:hyper-singular.boundary} also hold for the kernel $\mathcal{K}_s^{(R)}$. In this case the diagonal of the CPD weight function becomes
\begin{equation*}
\Omega^R(w,w) \DEF \frac{\gammafcn((s-1)/2)}{\sqrt{\pi}\gammafcn(s/2)} \left| 1 - \frac{w-w_*}{2R} \right|^{-1}.
\end{equation*}
In particular, we have the limit
\begin{equation*}
\lim_{R\to\infty} \Omega^R(z,w) = \frac{\gammafcn((s-1)/2)}{\sqrt{\pi}\gammafcn(s/2)}, \qquad s>\dim\Gamma(A),
\end{equation*}
where the convergence is uniform on compact subsets of $\Hplus\times\Hplus$. Consequently,
\begin{equation}
\frac{\mathcal{H}_d^{s-1,\Omega^R}}{\mathcal{H}_d^{s-1,\Omega^R}(A)} \weakstarto \frac{\mathcal{H}_d\big|_A}{\mathcal{H}_d(A)}, \qquad \text{as $R\to\infty$ and $s\geq\dim\Gamma(A)$.}
\end{equation}
Here $\mathcal{H}_d|_A/\mathcal{H}_d(A)$ is the limit distribution of minimal $\mathcal{K}_s^\infty$-energy $N$-point configurations as $N\to\infty$ (see next subsection). 
\end{rmk}

Of interest is the question of how well minimal $\mathcal{K}$-energy points are separated, that is, we are asking for a lower bound for the separation radius
\begin{equation}
\delta(Z_N^*) \DEF \min \left\{ \left| z - w \right| \mid z,w\in Z_N^*, z\neq w \right\}
\end{equation}
of optimal $\mathcal{K}$-energy $N$-point systems $Z_N^*$ valid for $N\geq2$. In fact, such an estimate can be obtained on sets of arbitrary Hausdorff dimension $\alpha$. The following result is a corollary of Theorem 4 in \cite{borodachov/hardin/saff:weighted}. 
\begin{prop} \label{prop:sep}
Let $0<\alpha<2$. Suppose $A$ is a compact subset in the interior of $\Hplus$ with $\mathcal{H}_\alpha(A)>0$. Let $\mathcal{K}=\mathcal{K}_s$ or $\mathcal{K}=\mathcal{K}_s^R$ with $R>0$. Then for every $s\geq\alpha$ there is a constant $c_s=c_s(A,\Omega,\alpha)>0$, where $\Omega$ is the CPD-weight function associated with $\mathcal{K}$, such that any $\mathcal{K}$-energy minimizing configuration $Z_N^*$ on $A$ satisfies the inequality
\begin{equation}
\delta(Z_N^*) \geq 
\begin{cases}
c_s N^{-1/\alpha} & s>\alpha, \\
c_\alpha (N\log N)^{-1/\alpha} & s=\alpha,
\end{cases}
\qquad N\geq2.
\end{equation}
\end{prop}

\subsection{The hypersingular case for the kernel $\mathcal{K}_s^\infty$} \label{hypersing.K.infty}

Suppose $s\geq1+\dim A$. The kernel $\mathcal{K}_s^\infty$ can be written as $\mathcal{K}_s^\infty(z,w)=\Omega^\infty(z,w) |z-w|^{1-s}$, where the CPD weight function $\Omega^\infty(z,w)=\gammafcn((s-1)/2)/[\sqrt{\pi}\gammafcn(s/2)]$ is a positive constant. Thus, we can apply the theory developed in \cite{borodachov/hardin/saff:weighted} to obtain

\begin{prop} \label{prop:hyper-singular.Kinfty}
Let $d=1$ or $d=2$. Suppose $A$ is a compact $d$-rectifiable set contained in the interior of $\Hplus$ with positive $d$-dimensional Hausdorff measure $\mathcal{H}_d(A)$. Let $s>1+\dim A$. For $N\geq2$, let $Z_N^*$ be a minimal $\mathcal{K}_s^\infty$-energy configuration of $N$ points $\{z_1^{s,N},\dots,z_N^{s,N}\}\subset A$. Then the sequence $Z_N^*$, $N\geq2$, is asymptotically uniformly distributed with respect to $\mathcal{H}_d$; that is, 
\begin{equation*}
\frac{1}{N}\sum_{k=1}^N \delta_{z_k^{s,N}} \weakstarto \frac{\mathcal{H}_d\big|_A}{\mathcal{H}_d(A)}, \qquad \text{as $N\to\infty$.}
\end{equation*}
Moreover, the minimal $N$-point $\mathcal{K}_s$-energy satisfies
\begin{equation} \label{weighted.energy.Kinfty}
\lim_{N\to\infty} \frac{\mathcal{E}_{\mathcal{K}_s^\infty}(A,N)}{N^{1+(s-1)/d}} = \frac{C_{s-1,d}^\infty}{\left[ \mathcal{H}_d(A) \right]^{(s-1)/d}},
\end{equation}
where $C_{s-1,d}^\infty$ is a positive constant which does not depend on $A$ and $N$. In fact, $C_{s-1,d}^\infty=C_{s-1,d}\gammafcn((s-1)/2)/[\sqrt{\pi}\gammafcn(s/2)]$ and $C_{s-1,d}$ is the same constant as in \eqref{Csd}.
\end{prop}

\begin{rmk}
The first part of Proposition \ref{prop:hyper-singular.Kinfty} holds for the boundary case $s=1+d$, $d=\dim A$, as well. (In the case $d=1$ it is also required that $A$ is contained in a $C^1$-curve.) The minimal $N$-point $\mathcal{K}_s$-energy satisfies
\begin{equation} \label{weighted.energy.boundary.case.K.infty}
\lim_{N\to\infty} \frac{\mathcal{E}_{\mathcal{K}_{d+1}}(A,N)}{N^2\log N} = \frac{\beta_d}{\mathcal{H}_d^{d}(A)},
\end{equation}
where $\beta_d=\pi^{d/2}/\gammafcn(1+d/2)$ is the volume of the unit ball in $\Rset^d$.
\end{rmk}

We remark that a separation result like Proposition \ref{prop:sep} can also be stated for $\mathcal{K}=\mathcal{K}_s^\infty$.

\appendix

\section{Convexity of the $\mathcal{K}_{s}$-kernel on vertical line-segments}
\label{appdx}

The parameter $v$ is fixed. The second derivative of \eqref{K.s.vertical-line} with respect to $y$ is  \begin{equation*}
\begin{split}
&s^{-1} \left( 4 R^{2} + \Delta^{2} \right)^{4+s/2} \frac{\dd^{2}}{\dd y^{2}} \mathcal{K}_{s}(R+\ii y,R+\ii v) \\
&\phantom{=}= - \left( 4 R^{2} + \Delta^{2} \right)^{2} \left[ 4 R^{2} - \left( 1 + s \right) \Delta^{2} \right] \Hypergeom{2}{1}{1/2,s/2}{1}{\frac{4 R^{2}}{4 R^{2} + \Delta^{2}}} \\
&\phantom{=\pm}- 2 R^{2} \left( 4 R^{2} + \Delta^{2} \right) \left[ 4 R^{2} - \left( 3 + 2 s \right) \Delta^{2} \right] \Hypergeom{2}{1}{3/2,1+s/2}{2}{\frac{4 R^{2}}{4 R^{2} + \Delta^{2}}} \\
&\phantom{=\pm}+ 6 R^{4} \left( 2 + s \right) \Delta^{2} \Hypergeom{2}{1}{5/2,2+s/2}{3}{\frac{4 R^{2}}{4 R^{2} + \Delta^{2}}},
\end{split}
\end{equation*}
where $\Delta$ denotes the difference $(y-v)$. Applying to each hypergeometric function the linear transformation \cite[15.3.3]{abramowitz/stegun:70} and simplifying we get
\begin{equation*}
\begin{split}
&s^{-1} \left( 4 R^{2} + \Delta^{2} \right)^{5/2} \left|\Delta\right|^{1+s} \frac{\dd^{2}}{\dd y^{2}} \mathcal{K}_{s}(R+\ii y,R+\ii v) \\
&\phantom{=}= \left( 1 + s \right) \Delta^{4} \Hypergeom{2}{1}{1/2,1-s/2}{1}{\frac{4 R^{2}}{4 R^{2} + \Delta^{2}}} \\
&\phantom{=\pm}+ 2 R^{2} \left[ 4 s R^{2} + \left( 1 + 2s \right) \Delta^{2} \right] \Hypergeom{2}{1}{1/2,1-s/2}{2}{\frac{4 R^{2}}{4 R^{2} + \Delta^{2}}},
\end{split}
\end{equation*}
which implies $\dd^{2} [\mathcal{K}_{s}(\gamma(y),\gamma(v))]/\dd y^{2} > 0$.
\providecommand{\bysame}{\leavevmode\hbox to3em{\hrulefill}\thinspace}
\providecommand{\MR}{\relax\ifhmode\unskip\space\fi MR }
\providecommand{\MRhref}[2]{%
  \href{http://www.ams.org/mathscinet-getitem?mr=#1}{#2}
}
\providecommand{\href}[2]{#2}

\end{document}